\documentclass{article}
\usepackage{graphicx}
\usepackage{amsfonts}
\usepackage{amsmath}
\usepackage{amssymb}
\usepackage{url}
\usepackage{bm}
\usepackage{overpic}

\usepackage{dsfont}
\usepackage{fancyhdr}
\usepackage{indentfirst}
\usepackage{enumerate}
\usepackage[normalem]{ulem}
\usepackage{mathrsfs}

\usepackage[colorlinks=true,citecolor=blue]{hyperref}
\usepackage{amsthm}
\usepackage{color}

\definecolor{DarkGreen}{rgb}{0.2,0.6,0.2}

\usepackage{natbib}
\usepackage{comment}

\def\d{\mathrm{d}}
\def\laweq{\buildrel \mathrm{d} \over =}

\newcommand{\E}{\mathbb{E}}

\newcommand{\R}{\mathbb{R}}
\newcommand{\RR}{\mathbb R\cup\{+\infty\}}

\newcommand{\N}{\mathbb{N}}
\newcommand{\p}{\mathbb{P}}
\newcommand{\X}{\mathcal{X}}
\newcommand{\Z}{\mathcal{Z}}

\newcommand{\G}{\mathcal{G}}

\newcommand{\id}{\mathds{1}}

\renewcommand{\ge}{\geqslant}
\renewcommand{\le}{\leqslant}

\renewcommand{\leq}{\leqslant}
\renewcommand{\epsilon}{\varepsilon}
\newcommand{\esssup}{\mathrm{ess\mbox{-}sup}}
\newcommand{\essinf}{\mathrm{ess\mbox{-}inf}}

\theoremstyle{plain}
\newtheorem{theorem}{Theorem}
\newtheorem{corollary}{Corollary}
\newtheorem{lemma}{Lemma}
\newtheorem{proposition}{Proposition}
\theoremstyle{definition}
\newtheorem{definition}{Definition}
\newtheorem{example}{Example}
\newtheorem{assumption}{Assumption}
\usepackage{setspace}

\theoremstyle{remark}
\newtheorem{remark}{Remark}

\newcommand{\VaR}{\mathrm{VaR}}

\newcommand{\ES}{\mathrm{ES}}

\usepackage{footmisc}
\begin{document}

\title{Robustness in the Optimization of Risk Measures}

\author{Paul Embrechts\thanks{RiskLab, Department of Mathematics, and ETH Risk Center,  ETH Zurich, 8092 Zurich, Switzerland. Email: \url{embrechts@math.ethz.ch}}\and Alexander Schied\thanks{Department of Statistics and Actuarial Science, University of Waterloo,   Waterloo, ON N2L3G1, Canada. Email: \url{aschied@uwaterloo.ca}}\and  Ruodu Wang\thanks{Department of Statistics and Actuarial Science, University of Waterloo,   Waterloo, ON N2L3G1, Canada.
Email: \url{wang@uwaterloo.ca}.
 }}

 \date{}

\maketitle

\begin{abstract}
We study issues of robustness in the context of Quantitative Risk Management and Optimization. We develop a general methodology for determining whether a given risk measurement related optimization problem is robust, which we call ``robustness against optimization". The new notion is studied for various classes of risk measures and expected utility and loss functions. Motivated by practical issues from financial regulation, special attention is given to the two most widely used risk measures in the industry, Value-at-Risk (VaR) and Expected Shortfall (ES). We establish that for a class of general optimization problems, VaR leads to non-robust optimizers whereas convex risk measures generally lead  to robust ones.  Our results offer extra insight on the   ongoing discussion about the comparative advantages of VaR and ES in banking and insurance regulation. Our notion of robustness is conceptually different from the field of robust optimization, to which some interesting links are derived.
%

   \end{abstract}

   \newpage

\begin{center}
 {When a measure becomes a target, it ceases to be a good measure}.
 \\
 \hfill -- \emph{Goodhart's law, paraphrased by \cite{S97}}
\end{center}

\section{Introduction}

The main focus of this paper is the study of robustness properties of optimization procedures within Quantitative Risk Management (QRM).
For this, we introduce a novel general framework, which at the same time is conceptually intuitive and mathematically challenging.
A key and, as we will highlight in the paper, novel question concerns the influence of the choice of the underlying objective on the resulting robustness properties in risk optimization.
 In particular, we
are interested in the two most popular regulatory risk measures,
the Value-at-Risk (VaR) and the Expected Shortfall (ES), and their robustness properties in the context of risk optimization.

In QRM,  the concept of robustness for risk measures is traditionally studied at the level of objective functionals without involving optimization problems; see \cite{CDS10}, \cite{KPH13}, \cite{KSZ14, KSZ17}, \cite{EWW15}, and the references therein.
In the literature on robust optimization (see e.g.~\cite{BTGN09}),  model uncertainty is typically incorporated through modifying the objective functional or the constraints.

The paper \cite{CDS10}  compares the qualitative robustness of VaR and coherent risk measures; the authors conclude that VaR is better in their context.
Some later papers, e.g.~\cite{EWW15} and \cite{KSZ14, KSZ17},  put the corresponding arguments into a different perspective, showing that ES also has certain desirable  robustness properties.
Both streams of research assumed that both VaR and, say, ES are applied to the same financial position. In reality, however, the regulatory choice of a particular risk measure creates certain incentives, just like any other aspect of regulation. 
These incentives become effective even before that risk measure has ever been applied in a risk management context. For instance, once a specific risk measure has been chosen, portfolios will be optimized with respect to that risk measure. Thus, in reality, VaR and ES will typically not  be applied to the same position, and so one cannot decouple the technical properties of a chosen risk measure from the incentives it creates. In other words,  a risk measure as a standalone function may be robust, but fail to have desirable robustness properties when this measure is used within an optimization context.

In our paper, we make 
a first attempt of taking the incentives created by the choice of a risk measure  into account when assessing the risk measure's robustness properties.  In doing so, we   arrive at  completely different, and perhaps somewhat surprising,  conclusions concerning robustness properties than the previous literature.  To briefly illustrate our ideas, suppose that a risk factor is represented by a random variable $X$ arising from a stochastic model (later in the paper, $X$ will denote a random vector).
An investor has to optimize her position according to the
best of her knowledge, and hence
we shall refer to $X$ as the best-of-knowledge model, and the true model, denoted by $Z$, is typically unknowable.
Ideally, a good model $X$ is statistically close to $Z$ in  a sense to be made clear later.
Based on the best-of-knowledge model $X$ and an objective functional $\rho$, an optimized position  is chosen as a function $g(X)$ of $X$. Whereas the position $g(X)$ may have a desirable objective value $\rho(g(X))$,  this does not guarantee that $\rho(g(Z))$ is also desirable if  $Z$ is ``slightly" different  from $X$. 
In the absence of a perfect model, which almost always is the case in financial applications, this issue becomes crucially important.

Our motivation can be informally illustrated by Figure~\ref{ES_and_VaR_Pareto5}, which describes a situation in which the investor faces a Pareto-distributed risk with unknown parameter $\theta$.\footnote{The Pareto($\theta$) distribution function $F$ with parameter $\theta >0$ is specified as $F(x)=1-x^{-\theta}$, $x \ge 1$.} The investor takes a model $X$ with an estimated parameter  $\widehat\theta$ and optimizes $\rho(g(X))$ over a certain class of functions $g$; details are explained in the caption of Figure \ref{ES_and_VaR_Pareto5}.  We compute the actual risk $\rho(g_X(Z))$ faced by the investor for $\rho$ chosen as $ \VaR$ at level $0.99$ and $ \ES$ at level $0.975$ as in \cite{BASEL16}, where $g_X$ denote the corresponding optimizing functions (unique in the case of VaR)  using the model $X$ and $\rho$.
As one can see from Figure~\ref{ES_and_VaR_Pareto5}, 
the \emph{perceived} risk value $\rho(g_X(X))$ (see Section \ref{sec:23} for its interpretation) is similar for $\rho$ being VaR and ES; this also holds for the \emph{actual} risk value $\rho(g_X(Z))$ as long as $\theta>\widehat\theta $, meaning that $X$ is an overestimate of the true loss $Z$. If, however,  the true loss is slightly underestimated, the actual VaR of $g_X(Z)$ is substantially higher than  $\rho(g_X(X))$, whereas the actual ES remains almost flat. The intuitive explanation of this phenomenon is that under VaR it is optimal to concentrate all tail risk on an event whose  probability is so small that it does not affect the measured risk. If, however, the true probability of that event is underestimated by the model, then the tail risk suddenly does become significant. 
Simulation results show that, if $\widehat \theta$ is computed from a maximum likelihood estimator
based on iid sample points for the a true parameter,
then the mean square error measured by the difference between $\rho(g_X(X))$ and $\rho(g_X(Z))$ is large  for the case of VaR,
and it is tiny  for the case of ES (see
Section~\ref{sec:simulation}); this sharp contrast is expected from Figure \ref{ES_and_VaR_Pareto5}.

\begin{figure}[h]
\centering
\begin{overpic}[width=9cm]{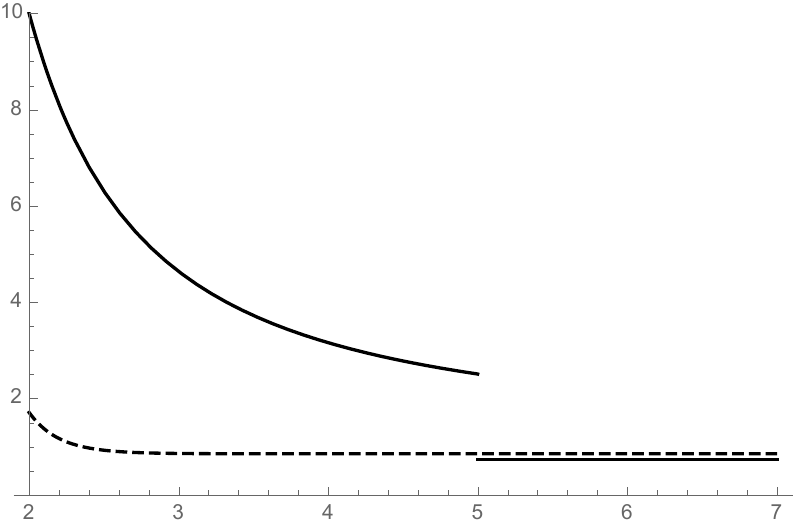}
\put(100,0){$\theta$}
\end{overpic}
\caption{
This plot shows $\rho(g_X(Z))$ for 
 $\rho=\VaR_{0.99}$ (solid) and $\rho=\ES_{0.975}$ (dashed), if $Z$ is Pareto distributed with parameter $\theta$. We assume that $X$ has a Pareto distribution with parameter $\widehat\theta=5$ and $g_X$ minimizes $\rho(g(X))$ within the class of all measurable functions $g$ 
satisfying  the inequality $0\le g(x)\le x$ for all $x\ge 0$ and the budget constraint $\mathbb E[\gamma(X)g(X)]\ge1$ formulated in Section \ref{sec:3}, where we take $\gamma(x)= x$ for simplicity.
The optimizer $g_X$ is unique in the case of VaR (Proposition \ref{prop:r1-var1}). See Section~\ref{sec:simulation} for more details and additional numerical simulations.}
\label{ES_and_VaR_Pareto5}
\end{figure}

In the present paper, we put the above observation into a rigorous quantitative framework  for general risk measures and optimization problems. 
The contribution and the structure of the paper are outlined below.
In Section \ref{sec:2}, we introduce the theoretical framework of robustness properties of risk measures in the context of optimization problems, referred to as ``robustness against optimization" throughout the paper.\footnote{We thank 
Paul Glasserman 
  for suggesting this terminology.}
The framework is quite general and it includes many practical problems in various fields of applications, not necessarily confined to finance and insurance.
Keeping the problem of risk measures in mind, a class of functional optimization problems is described in Section \ref{sec:3}. 
These optimization problems are analyzed in Sections \ref{sec:VaR}, \ref{rm section} and \ref{utility section} for VaR, convex risk measures, and expected loss (and utility)  functions, respectively. 
Robustness statements are obtained under general conditions, and further analytical solutions are available in some special cases. 
As the main message of our results, we see that, for the case of VaR which is argued by many as a robust risk measure, its corresponding optimization is highly non-robust and a small model uncertainty would ruin the optimality of the optimized positions. In sharp contrast, for many convex risk measures including ES and expected loss functions, the optimized positions are generally robust.  Some  simulation resuls, in addition to those underlying Figure \ref{ES_and_VaR_Pareto5}, are provided in Section~\ref{sec:simulation}. 
In Section \ref{sec:7}, we present some discussions on the implications of our results for the desirability of specific regulatory risk measures, an on-going debate in the financial industry (\cite{BASEL16}, \cite{IAIS14}).
Our results yield a (further) strong argument against using VaR as a risk measure within banking and insurance regulation;
for a related discussion on robustness in the realm of risk sharing, see \cite{ELW18}.
In the last section, Section \ref{sec:6a}, we discuss our notion of robustness in the context of distributionally robust optimization (e.g.~\cite{NPS08}, \cite{GS10}).
The proofs of all results are put in the appendix.


{As  Goodhart's law (\cite{G84}) implies,    when a risk measure becomes a target, it ceases to be a good risk measure.\footnote{\cite{D02} applied Goodhart's law   to  risk models by saying that a risk model breaks down when used for regulatory purposes.}
Nevertheless, to what extent this law applies depends on the specific application and each particular risk measure.  
Our results in this paper thus provide a quantitative and comparative analysis of Goodhart's law in the context of financial regulation and optimization.  
Combining our negative  result on VaR, which is an example of Goodhart's law,  and  our positive result for convex risk measures,  our  main message may be summarized as:
\begin{center}
\emph{As regulatory target, all risk measures cease to be good, but some risk measures, VaR in particular,  are much worse than the others.}
\end{center}}

%

 \section{Theoretical framework}\label{sec:2}
\subsection{Notation}
We work with an atomless probability space $(\Omega, \mathcal F, \p)$. Let $L^q$ be the set of all random variables in $(\Omega, \mathcal F, \p)$ with finite $q$-th moment, $q\in  {(0,\infty)}$, {$L^0$ be the space of all $\p$-a.s.~finite random variables, and $L^\infty$ the set of all essentially bounded elements of $L^0$}.
For a positive integer $n$, write $L^q_n=(L^q)^n$.
For a vector $x\in \R^n$, $|x|$ is its Euclidean norm.
Throughout, for any $X\in L^0$, $F_X$ represents the distribution function of $X$.
The mappings  $\essinf(\cdot)$ and $\esssup(\cdot)$ on $L^0$ stand for the essential infimum and the essential supremum of a random variable, respectively.
  We write $X\laweq Y$ if the random variables $X$ and $Y$   have the same distribution under $\p$.
For $x\in \R$, denote by $\delta_x$ the point-mass probability measure at $x$.
For real numbers or functions $x$ and $y$,  we write $ x\wedge y=\min\{x,y\}$, $x\vee y=\max\{x,y\}$, $x_+=x \vee 0$, and $x_-=(-x)\vee 0$.



 \subsection{Basic setup of optimization problems}
 
In this section, we first lay out the basic setup for optimization problems when the relevant information on the underlying economic model is known precisely, that is, the case without model uncertainty.
 Let $X$ be an $n$-dimensional random vector, where $n$ is a positive integer.
The random vector $X$ is called an \emph{economic vector}, which includes all random sources in an economic model under study,
such as potential losses, traded securities, hedging instruments, insurance contracts, macro economic factors, or pricing densities.

Let $\mathcal G_{n}$ be the set of {measurable} functions mapping $\R^{n}$ to $\R$.
A random variable $g(X)$ where $g\in \mathcal G_n$ represents a risky position of an investor, in which positive values represent a loss and negative values represent a gain. This sign convention is in line with the regulatory angle we follow in several applications in our paper.
The investor's problem is to choose among admissible  positions $g(X)$ for some functions $g$ in an admissible set $\mathcal G\subset \mathcal G_n$.

For  a set $\mathcal G\subset \mathcal G_n$,
we formulate the problem
\begin{equation}\mbox{to minimize: } \rho(g(X)) \mbox{~~~~~subject to~} g\in \mathcal G, \label{eq:opt1}\end{equation}
where $\rho$ is an objective functional mapping a set containing $\{g(X):g\in \mathcal G\}$ to $\RR$.
Here one prefers a smaller value of the objective functional over a larger value.
Objective functionals  considered {may be}  general; examples  include (up to a sign change) mean-variance functionals, expected utilities,  rank-dependent utility functionals, functionals in cumulative prospect theory, and various risk measures as discussed in \cite{ADEH99} and \cite{FS16}. {Our main interest will be, however, in the risk measures Value-at-Risk (VaR) and Expected Shortfall (ES\footnote{A formal definition of VaR and ES is given in Section \ref{sec:3}. ES is also known as CVaR, CTE, AVaR and TVaR, depending on the context (see e.g.~\cite{PR07}, \cite{MFE15} and \cite{FS16}). We use the term ES to be consistent with the Basel Committee on Banking Supervision (\cite{BASEL16}), because our study is motivated by the comparative advantages of VaR and ES in regulation; see discussions in Section \ref{sec:7}.}).}

The elements in the optimization problem \eqref{eq:opt1} can be summarized by an objective functional $\rho$ and a pair $(X,\mathcal G)$.
We always assume that the domain of the objective functional $\rho$ contains $\{g(X):g\in \mathcal G\}$, otherwise \eqref{eq:opt1} is meaningless.

\begin{example}\label{hedge ex}
An illustrative example is the classic problem of hedging in a financial market. Suppose that an investor currently faces a risk $W$ and would like to hedge against $W$.
She has access to hedging instruments in a set $\{g(Y):g\in \mathcal G'\}$ where $Y$ is an ${(n-1)}$-dimensional economic vector and, say, $\mathcal G'\subset \mathcal G_{{n-1}}$, ${n\ge 2}$.
Typically, the set $\mathcal G'$ involves a budget constraint.
Equivalently, she chooses risky positions in the set $\{W-g(Y): g\in \mathcal G'\}$, which represents all possible hedged positions she may attain.
Note that $\{W-g(Y): g\in \mathcal G'\}=\{f(X): f\in \mathcal G\}$ where $X=(W,Y)$ is an ${n}$-dimensional random vector and
 $$\mathcal G=\{f \in \mathcal G_{{n}}:f(w,y)=w-g(y),~g\in \G',~w\in \R,~y\in \R^{{n-1}}\};$$ therefore the hedging problem is a special case of our general setting \eqref{eq:opt1}.
Here we allow $W$ to be arbitrarily dependent on $Y$. If $W$ is a financial risk and $Y$ is the vector of asset prices in a complete financial market, then it may be that $W$ {is a function of} $Y$.
On the other hand, if $W$ represents a property and casualty insurance risk and $Y$ is  the vector of asset prices in a financial market, then it may be reasonable to assume that $W$ and $Y$ are independent.
\end{example}

 \begin{example}\label{Markowitz ex}In Markowitz's portfolio selection, an investor wishes to find an optimal allocation vector  $w\in\mathbb R^n$ based on a return vector $Y\in L^2_n$ for $n$ stocks. The problem can be described as the minimization of $\mathbb E[g(Y)]$ where $g\in\mathcal G'$ and $\mathcal G'$ consists of all functions $g\in \mathcal G_n$ that are of the form $g(y)=\lambda(w^\top y)^2-w^\top y$ for $w\in\mathbb R^n$ satisfying $\sum_{i=1}^nw_i=x_0$. Here, $\lambda>0$ is a risk-aversion parameter and $x_0$ represents the budget constraint of the investor.
\end{example}

In this section, we keep the choice of $(X,\mathcal G)$ as generic as possible. Special cases and examples are studied in Sections \ref{sec:3}-\ref{utility section}.

\subsection{Uncertainty and robustness against optimization}
\label{sec:23}

We proceed to put uncertainty into the optimization problem \eqref{eq:opt1} described above.
For  $X\in L_n^0$, $\mathcal G\subset \mathcal G_n$ and an objective functional $\rho$,
denote by $\rho(X;\mathcal G)$ the minimum possible value of $\rho$, namely,
$$\rho(X;\mathcal G)= \inf\{\rho(g(X)): g\in \mathcal G\},$$ 
and by $\mathcal G_X(\rho)$ the set of \emph{optimizing functions}, that is,
\begin{equation}\label{eq:opt2}
\mathcal G_X(\rho)=\{g\in \mathcal G:\rho(g(X))=\rho(X;\mathcal G)\}.
\end{equation}
Note that $\mathcal G_X(\rho)$ might be an empty set.  Throughout this paper, the notation $g_X$ will refer to a generic element $g_X\in  \mathcal G_X(\rho)$, and  
$g_X(X)$ will be called an \emph{optimized position}.

We shall use  $X$ to represent our (perceived) model  for the underlying economic vector.
In practice, the model $X$ is obtained based on stochastic assumptions  and statistical inference, and it may not represent a true model  for the underlying economic vector.
In other words, the  optimization problem \eqref{eq:opt1} is often subject to severe model uncertainty. To reflect this issue, let $\Z\subset L_n^0$ be a set of \emph{possible} economic vectors including $X$; $\Z$ may be interpreted as the set of \emph{alternative} models.\footnote{
 For instance, $\Z$ can be a parametric family of risk models and hence corresponds to parameter uncertainty of the underlying risk models. 
Each component of the vector $X$ may have a different economic meaning.  Some of them may be subject to more severe model uncertainty whereas others may be free of model uncertainty.
This can be reflected in the choice of $\Z$ which may be contained in a low-dimensional subset of the set of $n$-random vectors.}
 Suppose that the \emph{real} economic vector  $Z\in \Z$ is different from the \emph{perceived} economic vector $X$.
The information  we have at hand is  about $X$ rather than $Z$,
and we shall  refer to $X$ as the \emph{best-of-knowledge model} and $Z$ as the \emph{true model}, which is unknowable.
We have to make decisions according to the best of our knowledge, that is, as in \eqref{eq:opt2}, choosing $ g_X \in   \mathcal G_X(\rho)$  optimizing our objective $\rho$.
The  \emph{real} but unknown   position $  g_X (Z)$ may be different from the \emph{perceived} optimal 
 position $  g_X (X)$.
If $Z$ and $X$ are close to each other according to some {(pseudo-)}metric $\pi$ (e.g.~$L^\infty$-metric on the   space of bounded random vectors),
we would like  $\rho(  g_X (Z))$ to be close to $\rho(  g_X (X))$ in order to make sense of the position $  g_X (Z)$, which may no longer be optimal.
 In other words, we naturally would desire some continuity of the mapping  {$Y\mapsto \rho(  g_X (Y ))$} at $Y=X$.

 {Note that, in our situation, there is no point in analyzing  the problem of optimizing $g(Z)$ over $g\in \mathcal G$, because $Z$ is unknown.} This makes our framework conceptually different from the stream of research on stability of the set of optimizers under model uncertainty. In Section \ref{sec:remarks}, using examples from financial regulation,
  we make explicit this very important distinction between our paper and other approaches in the literature.

Putting this into the {hedging} context of Example \ref{hedge ex}, suppose that
the real economic vector $Z\in \Z$ is different from $X$ and
an investor has the real risk $W=h(Z)$ to hedge.
The information she has is  about $X$, and she hedges $W$ by choosing $g_X$ from a set of available instruments $\mathcal G'$.
In this case,   $h(Z)-g_X(Z)$   is the remaining risk she actually faces after hedging.
Under this setting, assuming the quality of model $X$ is good, $Z$ should be close to $X$ in some sense, and she would naturally desire some continuity of the function $ {Y\mapsto\rho(h(Y)-  g_X (Y))}$  at $Y=X$. 

The admissible set $\mathcal G$ is not subject to model uncertainty, as the investor knows which positions she can choose in the optimization problem.
For instance, in the above hedging example, a budget constraint that determines $\mathcal G$ is not affected by the model assumptions made for $X$; it is simply the observed prices for the hedging instruments.

In light of the above consideration,
we endow the set $\mathcal Z$ of all possible economic vectors with a pseudo-metric $\pi$. 
Common choices of $\pi$ are the $L^\infty$, the $L^q$, the Wasserstein, and the weak (pseudo-)metrics; see Example \ref{ex:metrics} below.
The reason for considering a pseudo-metric instead of a metric is to be able to incorporate  objective functionals based on the distributions of risks, for instance, law-invariant risk measures and expected utility functions.\footnote{Because of the extensive use of simulation and estimation methods for risk evaluation, uncertainty at the level of distributions is the most common in risk management practice; on the other hand, in optimization problems, risks with identical distributions are not equivalent (see our functional optimization problems in Section \ref{sec:3}).  For this reason, we do not use the equivalent class induced by the pseudo-metric.}
%


\begin{definition}\label{def:0}
We call $(\G,\Z,\pi)$ an \emph{uncertainty triplet} if  $\G\subset \G_n$ and $(\Z,\pi)$ is a pseudo-metric space of $n$-random vectors.
For a given uncertainty triplet $( \mathcal G,\Z,\pi)$, we say that an objective functional $\rho$ is \emph{compatible} if
$\rho$ maps $\G(\Z)=\mathcal\{g(Z):Z\in \Z,~g\in \mathcal G\}$ to $\RR$, and $\rho(g(Y))=\rho(g(Z))$ for all $g\in \mathcal G$ and  $Y,Z\in \Z$ with $\pi(Y,Z)=0$, i.e.~$Y$ and $Z$ are  indistinguishable under the pseudo-metric $\pi$.\end{definition}

\begin{definition}\label{def:1}
Let $( \mathcal G,\Z,\pi)$ be an uncertainty triplet.
A compatible objective functional
$\rho$ is \emph{robust against optimization   at
$X\in \Z$ for $(\G, \Z,\pi)$}  if there exists $  g_X \in  \mathcal G_X(\rho)$ such that
the function
$ {Y\mapsto \rho(  g_X (Y))}$ is $\pi$-continuous at $Y=X$.  
\end{definition}


In this paper, we are mainly interested in robustness  in the sense of Definition \ref{def:1}, and it should not be confused with the classic qualitative robustness of risk measures as studied in e.g.~\cite{CDS10}, \cite{KPH13}, \cite{KSZ14} and \cite{EWW15}.
On the other hand, in contrast to the robust optimization literature (e.g.~\cite{GS10}, \cite{WKS14}),
our focus is the robustness of objective functionals in optimization, instead of \emph{how to solve} particular optimization problems.
As such, our setup and methodology are also different from classic ones in the optimization literature.
  Section \ref{sec:remarks} contains  further discussions on the formulation of Definition \ref{def:1}, including some possible alternatives.


Below we give three prominent examples of $\pi$, which will appear throughout the paper.

\begin{example}\label{ex:metrics}
\begin{enumerate}[(i)]
\item For   a subset $\mathcal Z$  of $L^\infty_n$,
the $L^\infty$-metric $\pi_n^\infty$ is defined as
\begin{equation}\label{eq:linftym}
\pi_n^\infty(X,Y)=||X-Y||_\infty=\esssup(|X-Y|),~X,Y \in \mathcal Z.
\end{equation}
\item For $q\in [1,\infty)$ and a subset $\mathcal Z$  of $L^q_n$,
the $L^q$-metric $\pi_n^q$ is defined as
\begin{equation}\label{eq:lqm}
\pi_n^q(X,Y)=||X-Y||_q=(\E[|X-Y|^q])^{\frac 1q},~X,Y \in  \mathcal Z.
\end{equation}
\item For a subset $\Z$ of $L^0_n$,
the pseudo-metric $\pi^W_n$ is defined as
\begin{equation}\label{eq:weakm}
\pi_n^W(X,Y)=\pi_P(F_X,F_Y),~X,Y \in  \mathcal Z.
\end{equation}
where $\pi_P$ is the Prokhorov metric over the set of probability distribution measures.\footnote{Precisely, $\pi_P(\mu ,\nu )=\inf \left\{\varepsilon >0:\mu (A)\leq \nu (A_{{\varepsilon }})+\varepsilon \ {\text{and}}\ \nu (A)\leq \mu (A_{{\varepsilon }})+\varepsilon \ {\text{for all}}\ A\in {\mathcal  {B}}(\R^n)\right\}$, where $A_{{\varepsilon }}=\{x\in \R^n:  ||x-y||<\varepsilon \mbox{~for some~}y\in A\}$ and $||\cdot||$ is the Euclidean norm.
}
In this case, convergence in $\pi_n^W$ is equivalent to convergence in distribution (or weak convergence). 
One could  also replace $\pi_P$ in \eqref{eq:weakm} by a Wasserstein metric over probability measures, and obtain a Wasserstein pseudo-metric on a suitable subspace of $L^0_n$.
\end{enumerate}
\end{example}

\subsection{Basic properties on {robust}ness and continuity}\label{sec:24}

We first explain a few basic properties about Definition \ref{def:1}.
Robustness against optimization is a  {joint property} of {$(\rho, X,\G, \Z,\pi)$}, and only a {$\pi$-neighbourhood} of $X$ in $\Z$ matters in the definition.
If $\rho$  is  robust against optimization at
$X$ for $(\G, \Z,\pi)$, then  $\rho$  is   also  robust against optimization  at $X$ for $(\G,   \Z',\pi)$ if   {$X\in \Z'\subset \Z$},
and the same holds true for $(\G, \Z, \pi')$ if $ \pi'$ is a {stronger}  pseudo-metric than $\pi$.
On the other hand, if the optimization problem does not admit a solution, that is, $ \mathcal G_X(\rho) =\varnothing$, then  $\rho$ is not {robust against optimization} at $X$.

Robustness of $\rho$ relies on both some continuity of $\rho$ on $\G(\Z)$ and some continuity of functions in $ \mathcal G_X(\rho)$.  
In what follows, we give a few general results where  $ \mathcal G_X(\rho)$ contains a continuous function.
Whereas these results are fairly simple, they could nevertheless be useful in   situations where $\mathcal G$ is nice enough.  However, later we will see that in many representative problems, $ \mathcal G_X(\rho)$ does not necessarily have any continuous elements for commonly used risk measures such as VaR and convex risk measures; their robustness properties will be investigated in Sections \ref{sec:VaR} and \ref{rm section}.

For a bijection $g\in \mathcal G_n$ and a pseudo-metric space $(\Z,\pi)$ of $n$-dimensional random vectors,
let $(g(\Z),\pi_{g})$ be another pseudo-metric space defined as
$$\pi_{g}(g(X),g(Y))= \pi(X,Y)~~ \mbox{for $X,Y\in \Z$.}$$

\begin{proposition}\label{prop:31}
Suppose that for an uncertainty triplet $(\mathcal G,\Z,\pi)$, $X\in \Z$ and a compatible objective functional $\rho$, $ \mathcal G_X(\rho)$ contains a bijection $g$, and $\rho$ is $\pi_{g}$-continuous on $g(\Z)$. Then
$\rho$ is robust against optimization at
$X$ for $(\G, \Z,\pi)$.
\end{proposition}

Next we look at the basic settings of $\Z=L^\infty_n$, $\Z=L^q_n$, and $\Z=L^0_n$, equipped with the $L^\infty$ metric $\pi_n^\infty$ in \eqref{eq:linftym}, the $L^q$ metric $\pi_n^q$ in \eqref{eq:lqm}, and the pseudo-metric $\pi_n^W$ in \eqref{eq:weakm}, respectively.

\begin{proposition}\label{prop:32}
Let   $\rho$ be a compatible objective functional for the uncertainty triplet $(\mathcal G,\Z, \pi)$ and $X\in \Z$.
\begin{enumerate}[(i)]
\item Suppose $(\Z,\pi)=(L^\infty_n, \pi^\infty_n)$. If
 $ \mathcal G_X(\rho)$ contains a continuous function $g:\R^n\to \R$ and $\rho$ is $\pi_1^\infty$-continuous, then
$\rho$ is robust against optimization at
$X$ for $(\G, \Z, \pi)$.
\item Suppose $(\Z,\pi)=(L^q_n, \pi^q_n)$, $q\in [1,\infty)$.
If  $ \mathcal G_X(\rho)$ contains a  continuous and linearly growing\footnote{A function $g:\R^n\to \R$ is \emph{linearly growing}
if for some $C>0$, $|g(y)|\le C|y|$  for all  $y\in \R^n$ with $|y|>1$. This property is satisfied by, for instance, Lipschitz-continuous functions. } function $g:\R^n\to \R$ and  $\rho$ is $\pi_1^q$-continuous, then
$\rho$ is robust against optimization at
$X$ for $(\G, \Z,\pi)$.
\item Suppose $(\Z,\pi)=(L^0_n, \pi_n^W)$. If
 $ \mathcal G_X(\rho)$ contains a continuous function $g:\R^n\to \R$ and $\rho$ is $\pi_1^W$-continuous, then
$\rho$ is robust against optimization at
$X$ for $(\G, \Z, \pi)$.
\end{enumerate}
\end{proposition}

Proposition \ref{prop:32} provides simple criteria for verifying robustness of some objective functionals based on continuity of the optimizing functions in $\mathcal G_X(\rho)$.
As we shall see in Sections \ref{sec:VaR}-\ref{rm section}, for the popular risk measures VaR and ES, such criteria may not be very useful, as typically the optimizing functions lack the corresponding continuity.
More detailed analyses are needed to draw meaningful  conclusions for these objectives,
which will be the focus of  the next few sections.

\section{A class of functional optimization problems}\label{sec:3}


Our main {interest is in the robustness of risk measures} in optimization, and in particular, Value-at-Risk (VaR)  and    Expected Shortfall (ES).
Here, a positive value of $Y$ represents a loss and a negative value represents a gain.
The VaR at confidence level $p\in(0,1)$ is defined as
$$
 \VaR_p(Y)= \inf\{x\in \R: \p(Y\le x)\ge p\}=F_{Y}^{-1}(p),~~Y\in L^0,
$$
and the ES at confidence level $p\in(0,1)$ is defined as
\begin{equation}\label{ES def}
  \ES_p(Y)=\frac{1}{1-p}\int_p^1 \mathrm{VaR}  _s(Y)\d s,~~Y\in L^0.
\end{equation}
Note that $\ES_p(Y)$ may take the value $\infty$ if $Y$ is not integrable.
In addition, we  write
$$\ES_1(Y)=\VaR_1(Y)= \esssup(Y)= \sup\{x\in \R: \p(Y\le x)<1\}.$$

We summarize  {some well known} robustness properties of $\VaR_p$ and $\ES_p$ for $p\in (0,1)$ below.
\begin{enumerate}
\item[1.)]
$\VaR_p$ is continuous with respect to convergence in distribution, and hence (Wasserstein) $L^q$-convergence for $q\in [1,\infty]$, at $X$ if and only if the inverse cdf of $X$ is continuous at $p$; see, e.g., Proposition 7.3.1 in \cite{Shorack}.
\item[2.)]  It follows immediately from 1.) and \eqref{ES def} that $\ES_p$ is continuous with respect to convergence in distribution on every uniformly integrable subset of $ L^1$. In particular, $\ES_p$ is continuous with respect to   $L^q$-convergence for $q\in [1,\infty]$. On the other hand, $\ES_p$ is not continuous with respect to convergence in distribution on any set containing $L^\infty_+$.
\end{enumerate}
In addition to VaR and ES, we will consider two general classes of convex risk measures, 
as well as expected utility and loss functions, which will be introduced in Sections
\ref{rm section} and \ref{utility section}.

Next, we describe a general class of functional optimization problems. 
 %
%
%
%
For an $n$-dimensional random vector  $X$, two measurable functions
$v:\mathbb R^n\to \mathbb R\cup\{-\infty\}$ and $ w:\mathbb R^n\to \mathbb R$,  a measurable price density $\gamma:\mathbb R^n\to(0,\infty)$, and a constant $x_0\in\mathbb R$,
we consider the following set   
\begin{equation}
\label{eq:general-g}
\mathcal G=\Big\{g\in\mathcal G_n: v\le g\le w \mbox{ and }\mathbb E[g(X)\gamma(X)]\ge x_0 \Big\}.
\end{equation}
The corresponding optimization problem is
\begin{equation}
\label{eq:general-opt}
\mbox{to minimize: } \rho(g(X)) \mbox{~~~~~subject to~} 
v\le g\le w,~~ \E[ \gamma(X)g(X)]\ge x_0.
 \end{equation} 
 The optimization problem \eqref{eq:general-opt} is called `functional' because the objective is optimized over a large set of $n$-variate functions.

%
Intuitively, the functions $v$ and $w$   describe lower and upper bounds  on admissible functions $g\in\mathcal G$, while the condition $\mathbb E[g(X)\gamma(X)]\ge x_0 $ describes a budget constraint.\footnote{Using our sign convention,  by holding a risky position $g(X)$, one receives the monetary amount $\E[\gamma (X) g(X)]$. Equivalently, one pays $-\E[\gamma (X) g(X)]\le -x_0$, thus the usual budget constraint. 
}  
Hence, the  problem \eqref{eq:general-opt} represents   portfolio optimization with given budget, which is a classic problem in quantitative finance, and includes many interesting special cases. 
Below we present two  simple examples, one in the context of a complete financial market and the other one in the context of   insurance design. 
Since this problem has attracted substantial interest in its own right,   we will en passant contribute to the corresponding literature; see Remark \ref{opt remark}.

%
 
 \begin{example}[Optimal investment]\label{ex:ctfinmarket}
The optimization problem  \eqref{eq:general-opt}  
connects as follows with  continuous-time optimal hedging problems in a complete market with an arbitrary number of primary assets. Suppose that $S=(S_t)_{t\in[0,T]}$ is a $d$-dimensional semimartingale admitting a unique local martingale measure $\mathbb Q$ with price density $\gamma=\d\mathbb Q/\d\mathbb P$ on $\mathcal F=\mathcal F^S_T$.  
We can interpret $S$ as the discounted price process of $d$ risky securities. The fact that $\mathbb Q$ is unique is equivalent to the completeness of the model. Let $X$ be a random vector with $\sigma(X)=\mathcal F^S_T$, which represents   market randomness (such a random vector exists under mild conditions such as continuity of $S$). 
In this context, for an investor who needs to pay a random wealth $f(X)$ at time $T$, it is a natural task to minimize, e.g.,  $\rho(f(X) -V_T)$, where  $V_T:=V_T(X)$ is the discounted time-$T$ value of a self-financing trading strategy satisfying a cost constraint on the initial investment $V_0$ and   certain other constraints such as $v(X) \le V_T \le w(X)$ for some functions $v$ and $w$. By martingale arguments, the initial investment satisfies $V_0=\mathbb E[\gamma (X)V_T]$. On the other hand, market completeness implies that $S$ has the martingale representation property, and so every feasible functional $g(X)$ can be represented in the form $f(X)-V_T$, where $V$ is again the value process of some self-financing trading strategy. 
So, by letting $g(x):=f(x)-V_T(x)$, we arrive at a special case of \eqref{eq:general-g}. 
\end{example}

\begin{example}[Insurance design]
\label{ex:insurance} 
The problems of insurance design, pioneered by \cite{A63}, can also be described by \eqref{eq:general-opt}.  
Let $X\ge 0$ represent a random future loss to an insured, and $f(x)$ represent the 
amount of payment from the insurer if the realized loss $X$ is equal to $x$; $f$ is called an insurance indemnity function.
Suppose $\gamma \ge 1$ is a constant, and $\gamma \E[ f(X)]$ is used to price the insurance contract with payment $f$ (premium based on the expected payment is called the actuarial premium). Let $y_0$ be the budget of the insured. 
The standard optimal insurance problem of the insured with risk measure $\rho$ is 
 $$
\mbox{to minimize: } \rho(X-f(X)) \mbox{~~~~~subject to~} 
0 \le f(X)\le X,~~\gamma \E[ f(X)]\le y_0,
 $$
which belongs to \eqref{eq:general-opt} by choosing $g(x)=x-f(x)$, $w(x)=x$, $v(x)=0$, and
 $x_0=\gamma \E[ X]- y_0$. Some other requirements on the payment $f$ may be further imposed; see e.g.,   \cite{BHYZ15} and the references therein.
\end{example}

To study  the robustness of risk measures for this problem, the main assumption on $(\mathcal G,\mathcal Z,\pi)$ and $X$  is given below, which will be assumed in  the next three sections.  
The assumption on $X$ is standard and satisfied by practically all financial models in which the assets prices  have densities. 
\renewcommand\theassumption{G} 
\begin{assumption}\label{assm:g}
$\mathcal G$ is given by \eqref{eq:general-g} where $\E[\gamma(X)]<\infty$  and $\mathcal G\ne \varnothing$;
 the distribution measure $\mu_X$  of $X$ has a positive density on its support, which is a convex subset of $\R^n$, and 
$(\mathcal Z,\pi)$ is  $(L_n^0, \pi^W_n)$ or $(L_n^q, \pi^q_n)$, $q\in [1,\infty]$.   
\end{assumption}

In general, it is difficult to obtain analytical solutions to \eqref{eq:general-opt}; 
instead, we will obtain robustness statements on various risk measures in the subsequent sections. 
Explicit optimizers can be obtained for a particular case
of \eqref{eq:general-opt} such that $X$ is one-dimensional,
$v(x)=0$ and $w(x)=x$  (see e.g., Example \ref{ex:insurance});
that is, 
\begin{equation}
\label{eq:r1-var1}
\mbox{to minimize: } \rho(g(X)) \mbox{~~~~~subject to~} 
0 \le g(X)\le X,~~\E[ \gamma (X) g(X)]\ge x_0.
 \end{equation} 
 
%
%
%

\section{Robustness of Value-at-Risk}\label{sec:VaR}

The main task of this section is to establish the (non-)robustness of VaR for the optimization problem \eqref{eq:general-opt} in Section \ref{sec:3}.   
We will make the following assumption on $\gamma$, $v$ and $w$, as well as  the minimum value of the risk measure $\rho=\VaR_p$ for $p\in (0,1)$. Recall that $\rho(X;\mathcal G)= \inf\{\rho(g(X)): g\in \mathcal G\}$.
\renewcommand\theassumption{V}
\begin{assumption}\label{assm:v}
 $\esssup (v)< \rho(X;\mathcal G) < \rho(w(X))$ and $\gamma$ is bounded from above.
\end{assumption}

Assumption  \ref{assm:v} is quite general and weak.
The condition $  \esssup (v)< \rho(X;\mathcal G)$, meaning that the lower bound $v$ is not too large, is  primarily  used to guarantee that the budget constraint is binding. The condition $\rho(X;\mathcal G) < \rho(w(X))$ simply says that the optimization problem is not solved by trivially choosing the largest possible position $g=w$.
The boundedness condition on $\gamma$  may be significantly relaxed as we only need boundedness in a neighbourhood of one specific point. 
An explanation and a technical discussion on these conditions are put in Remark \ref{rem:var-assm} in Appendix \ref{app:r1-1}.

\begin{theorem}\label{th:var1}
For  $p\in(0,1)$, under Assumptions \ref{assm:g} and \ref{assm:v},
$\rho=\VaR_p$ is not {robust} against optimization at $X$ for  $(\G,  \mathcal Z , \pi)$.
\end{theorem}

Theorem \ref{th:var1} implies that, for the optimization problem  \eqref{eq:general-opt}  and all choices of commonly used (pseudo-)metrics, $\VaR_p$ is not {robust} against optimization, and this result holds for a general continuously distributed random vector $X$.  As a consequence, $\VaR_p$ has the poorest possible robustness in our setup.
The main reason behind this phenomenon is quite intuitive: as a key point in the proof of Theorem \ref{th:var1}, any optimizing function  $ g_X $ always has a jump at the $p$-quantile of $ {g_X(X)}$,  making it   most vulnerable to model uncertainty.

In practice, one may considers a subset $\mathcal G'$ of $\mathcal G$ which contains only continuous functions, so that robustness holds by Proposition \ref{prop:32}.
We emphasize that
the optimization for VaR is still problematic in this setting: One needs to search for  functions in $\mathcal G'$ which closely approximate the discontinuous function $g_X$, and hence its stability is still weak.

Next, we consider the   one-dimensional setting \eqref{eq:r1-var1}, which corresponds to $v(x)=0$ and $w(x)=x$.  
Denote by 
$
q:=\VaR_p(X;\mathcal G)
$  the minimum value of \eqref{eq:r1-var1} with $\rho=\VaR_p$,  
and we have $q>0$ by Assumption \ref{assm:v}.
The next proposition gives an explicit solution to \eqref{eq:r1-var1}, under a further minor condition that 
$(X-q)\gamma(X)$ has a unique $p$-quantile. 
%
\begin{proposition}\label{prop:r1-var1}
Let $p\in (0,1)$, $\rho=\VaR_p$ and $Y =(X-q)\gamma(X)$. Suppose that Assumptions \ref{assm:g} and \ref{assm:v} hold, $\E[\gamma (X) X]<\infty$, and
 $ \p  ({ Y }\le \VaR_p( Y)   )=p.$
 Problem \eqref{eq:r1-var1}  admits a $\mu_X$-a.s.~unique solution that is of the form
  \begin{equation} \label{eq:r1-varopt4-ns}
 {g_X(x)}=x\id_{\{{(x-q)\gamma(x)}> c \}}  + (x\wedge q) \id_{\{  (x-q)\gamma(x)\le  c \}},
\end{equation}
where $c=\VaR_p({ Y })$. Moreover,  
\begin{equation}\label{eq:r1-var3}
p\ES_{1-p}(-Y_+)= x_0- \E[\gamma(X) X].
\end{equation}
\end{proposition}
Since $Y_+= \gamma(X)(X-q)_+$,  the left-hand side of \eqref{eq:r1-var3} is an increasing function of $q$, and hence
the value $q$ can be  numerically computed by solving \eqref{eq:r1-var3}.
Section \ref{sec:simulation} contains simulation studies for the problem setting in Proposition \ref{prop:r1-var1}.

\section{Robustness of two classes of convex risk measures}\label{rm section}

In this section, we  obtain positive robustness results  for two important classes of convex risk measures, the \emph{divergence risk measures} and the \emph{utility-based shortfall risk measures}. The first class of risk measures, which are sometimes also called optimized certainty equivalents, contains Expected Shortfall as an important special case. The second class comprises the expectiles. Both classes intersect at the entropic risk measure, which can  also be treated within the setting of our subsequent Section  \ref{utility section}, where we will analyze the robustness of expected utility/expected loss. 

We continue to study the optimization problem  \eqref{eq:general-opt}.
The following simple regularity conditions on $\gamma$, $v$ and $w$ are made to establish results in this section.

\renewcommand\theassumption{P}
\begin{assumption}\label{assm:gamma}
The price density  $\gamma:\mathbb R^n\to(0,\infty)$ is $\mu_X$-a.e.~continuous  and $\gamma(X)$ has a continuous density. 
\end{assumption}

\renewcommand\theassumption{R}
\begin{assumption}\label{assm:vw}
The functions 
 $v$ and  $w$ are  $\mu_X$-a.e.~continuous.  Moreover,     $-\infty\le \mathbb E[\gamma(X)v(X)]\le x_0\le \mathbb E[\gamma(X)w(X)] \le \mathbb E[|\gamma(X) w(X)|]<\infty$.  
\end{assumption}


We consider a convex risk measure, i.e., a functional $\rho:L^1\to \mathbb R\cup\{+\infty\}$ satisfying monotonicity, cash invariance, and convexity (see, e.g., Chapter 4 in \cite{FS16}). 
Let $\varphi:\mathbb R\to[0,+\infty]$ be a proper closed convex function whose effective domain is an interval with endpoints $a<b$.  We assume moreover that $a<1<b$ and that $0=\varphi(1)=\min_x\varphi(x)$. Then the $\varphi$-divergence of a probability measure $\mathbb Q$ with respect to $\mathbb P$ is 
\begin{equation}
I_\varphi(\mathbb Q|\mathbb P):=\begin{cases}\displaystyle \int\varphi\Big(\frac{\d \mathbb Q}{\d\mathbb P}\Big)\d\mathbb P&\text{if $\mathbb Q\ll\mathbb P$,}\\+\infty&\text{otherwise.}
\end{cases}
\end{equation}
The corresponding divergence risk measure is defined as 
\begin{equation}\label{divergence risk measure rep eq}
\rho(Y):=\sup_{\mathbb Q\ll\mathbb P}\big(\mathbb E_{\mathbb Q}[Y]-I_\varphi(\mathbb Q|\mathbb P)\big),\qquad Y\in L^\infty.
\end{equation}
If $\varphi(x)=x\log x-x+1$, then $I_\varphi(\mathbb Q|\mathbb P)$ is the relative entropy, or the Kullback--Leibler divergence, of $\mathbb Q$ with respect to $\mathbb P$ and $\rho$ is an entropic risk measure. If $\varphi=\infty\cdot\id_{[1/(1-p),\infty)}$ for some $p\in[0,1)$, then $\rho$ is the Expected Shortfall, $\ES_p$.

\begin{theorem}\label{th:divergence risk measures} 
In addition to Assumptions \ref{assm:g}, \ref{assm:gamma}, and \ref{assm:vw} we assume that $v$ and $w$ are bounded. 
Then the divergence risk measure $\rho$ is robust against optimization at $X\in L^0_n$ for $(\mathcal G,L^0_n, \pi^W_n)$. \end{theorem}

The proof of Theorem \ref{th:divergence risk measures} relies on the duality formula from \cite{BenTalTeboulle1,BenTalTeboulle2}, which for general divergence risk measures has only been established on $L^\infty$.  This is one of the reasons for assuming the boundedness of  $v$ and $w$ in Theorem \ref{th:divergence risk measures}. It is  possible, however, to relax their boundedness  by imposing a suitable growth condition, given that the aforementioned duality formula extends to $L^p$ (or, more generally, to a certain Orlicz space). An important special case for which this is possible is Expected Shortfall, $\ES_p$. Recall that the Expected Shortfall   has the following representation (e.g., \citet[Theorem 8.14]{MFE15}) for $Y\in L^1$,
\begin{align}\label{ES representations eq}
\ES_p(Y)&=\frac1{1-p}\int_{p}^1F_Y^{-1}(s)\d s=\sup\Big\{\mathbb E_{\mathbb Q}[Y]:\mathbb Q\ll \mathbb P\text{ and }\frac{\d\mathbb Q}{\d\mathbb P}\le\frac1{1-p}\Big\}.
\end{align}
Note that the right-hand representation coincides with \eqref{divergence risk measure rep eq}
 if $Y\in L^\infty$ and $\varphi=\infty\cdot\id_{[1/(1-p),\infty)}$.
We will say that a function $f:\mathbb R^k\to\mathbb R$ has growth index $q\in[0,\infty]$,  if $f$ is locally bounded for $q=\infty$ and if for $q<\infty$ there exists a  constant  $c$ such that $|f(x)|\le c(1+|x|^q)$ for $x\in\mathbb R^k$.

\begin{corollary}\label{cor:ES} 
In addition to Assumptions \ref{assm:g}, \ref{assm:gamma}, and \ref{assm:vw}, we assume that both $v$ and $w$ have growth index $q\in[1,\infty]$. Then
Expected Shortfall, $\ES_p$, with $p\in(0,1)$,  is robust against optimization at $X\in L^q_n$ for $(\mathcal G,L^q_n, \pi^q_n)$. \end{corollary}

Comparing Theorem \ref{th:var1} and Corollary \ref{cor:ES}, we see that ES has clear advantages over VaR in terms of robustness against optimization.  
Note that Assumptions \ref{assm:g} and \ref{assm:v}   in Section \ref{sec:VaR} and  Assumptions \ref{assm:gamma}  and \ref{assm:vw} in this section are simple regularity conditions, and they are realistic for practical models. 
Therefore, we safely can say that for Problem \eqref{eq:general-opt},  VaR is generally not robust against optimization, and ES is generally robust against optimization. 

In Theorems \ref{th:var1}-\ref{th:divergence risk measures} and Corollary  \ref{cor:ES}, we assumed that $w$ is finite, and thus the admissible positions in the optimization problem \eqref{eq:general-opt}  has an upper bound. 
A similar comparison on the robustness of VaR and ES is obtained for the unbounded problem ($w=\infty$ and $v=-\infty$), which is presented in Appendix \ref{app:unbounded}.


Now we turn to an analysis of the robustness of utility-based shortfall risk measures as introduced in \cite{FS02}. To this end, let  $\ell:\R\to \mathbb R$ be a nonconstant, increasing, and   convex loss  function and $x_0$ be an interior point in the range of $\ell$. The corresponding utility-based shortfall risk measure is given by
$$\rho(Y)=\inf \big \{m\in\mathbb R: \E[\,\ell(Y-m)]\le x_0 \big \},\qquad Y\in L^\infty.
$$

\begin{theorem}\label{th:usr risk measures} 
In addition to Assumptions \ref{assm:g}, \ref{assm:gamma}, and \ref{assm:vw} we assume that $v$ and $w$ are bounded. 
Then 
the utility-based shortfall risk measure $\rho$ is robust against optimization at $X\in L^0_n$ for $(\mathcal G,L^0_n, \pi^W_n)$.
\end{theorem}

A notable special case of a utility-based shortfall risk measure is the expectile of $Y\in L^1$ at level $\tau\in[0,1]$, defined as the unique solution to the equation
$$\tau \mathbb E[\, (Y-z)_+ \,]=(1-\tau)\mathbb E[\, (Y-z)_-\,].
$$
Expectiles were introduced by \cite{NP1987} and have recently gained attention in the context of the discussion on backtesting risk estimates.
As stated\footnote{Note the differing sign convention in \cite{FS16}.} in Exercise 4.9.2.(b) of \cite{FS16}, for $Y\in L^\infty$ and $\tau\in(1/2,1]$, the expectile is equal to the utility-based shortfall risk for the loss function $\ell(x)=\tau x_+-(1-\tau) x_-$. For the special case of an expectile, the following corollary relaxes the boundedness condition by a more general growth condition, just as we did for Expected Shortfall in Corollary \ref{cor:ES}.

\begin{corollary}\label{cor:expectile} In addition to Assumptions \ref{assm:g}, \ref{assm:gamma}, and \ref{assm:vw}, we assume that both $v$ and $w$ have growth index $q\in[1,\infty]$. Then
the expectile at level $\tau\in(1/2,1]$  is robust against optimization at $X\in L^q_n$ for $(\mathcal G,L^q_n, \pi^q_n)$.
\end{corollary}

\begin{remark}\label{opt remark} The results presented in this section rely on obtaining concrete solutions to the problem of minimizing $\rho(g(X))$ over $g\in\mathcal G$. For constant constraint functions, $v$ and $w$, and a coherent risk measure $\rho$, this problem can be formulated as a composite hypothesis testing problem, and there exists a significant amount of corresponding literature; see, e.g., Sections 3.5, 8.3 and the corresponding bibliographical notes in \cite{FS16} for a summary. Much fewer results were obtained for the case of  nonconstant constraint functions, and those that are available often lack some concreteness. A notable exception is \cite{Sekine04}  in a one-dimensional setting which solves \eqref{eq:r1-var1}; see Proposition \ref{pr:Sekine lemma1} below. It is therefore worth pointing out that our proofs also provide structure results for the solutions of our optimization problems. Specifically, in the context of  Corollary \ref{cor:ES}, our proof yields that 
there exists a minimizer $g_X$ that has one of the following two forms, where $z^*\in\mathbb R$ and $c>0$ are suitable constants:
$$g_X(x)=\begin{cases}
(v(x)\vee z^*\wedge w(x))\id_{\{0<c\gamma(x)<1\}}&\text{or}\\
(v(x)\vee z^*\wedge w(x))\id_{\{c\gamma(x)>1\}}+v(x)\id_{\{c\gamma(x)\le 1\}}
\end{cases}
$$
Moreover, in the settings of  Theorem \ref{th:divergence risk measures} and of Theorem \ref{th:usr risk measures}, there exists a constant $z^*$ such that a minimizer $g_X$ is also a minimizer of the expected loss $\mathbb E[\ell(g(X)-z^*)]$ over $g\in\mathcal G$, where $\ell(x):=\sup_{y\ge0}(xy-\varphi(y))$ for Theorem \ref{th:divergence risk measures}. The problem of minimizing this expected loss and its robustness properties will be discussed in the subsequent Section \ref{utility section}.
\end{remark}

Finally, we consider the special setting 
\eqref{eq:r1-var1},  which corresponds to $v(x)=0$ and $w(x)=x$,
for $\rho=\ES_p$, $p\in (0,1)$.
This problem has an explicit solution based on Theorem 8.26 of \cite{FS16}, which is a slight generalization of a result by \cite{Sekine04}. 

\begin{proposition}\label{pr:Sekine lemma1}
Let $p\in (0,1)$ and $\rho=\ES_p$. 
Suppose that Assumption \ref{assm:gamma} holds and $0\le x_0<\E[\gamma (X) X]$. 
There exist constants $d>0$ and $r\ge0$ such that the function
\begin{equation}\label{eq:es-ns-1}
  g_X (x)=x\id_{\{\gamma(x)>d\}}+(x\wedge r)\id_{\{\gamma(x) \le d\}} , ~~x\in \R,
\end{equation}
solves Problem \eqref{eq:r1-var1}. Moreover, $r$ is a $p$-quantile of $ {g_X(X)}$. 
\end{proposition}

For given $X$, $\gamma$, and $r$, one can compute $d$ as a function of $r$  by numerically solving the equation $\E[\gamma (X) g_X(X)]=x_0$. Subsequently, one can find the optimal $r$ by numerically minimizing the expression $\ES_p(  g_X(X))$; see Section \ref{sec:simulation} for its implementation in  simulation studies. 


\section{Robustness of expected utility and loss functions}\label{utility section}

An expected loss $\rho_\ell $ is a mapping $$Y\mapsto\rho_\ell(Y)= \E[\ell(Y)],$$
where  $\ell:\R\to \mathbb R$ is a nonconstant, nondecreasing, and   convex  function.  
Up to a sign change and a constant shift, 
minimizing an expected loss is equivalent to maximizing an expected utility, via the relation $\ell(x)=-u(b-x)$, where $u$ is a concave utility function, and  $b$ is the constant wealth level of the decision maker.
As in Section \ref{rm section}, we consider the set $ \mathcal G$ in \eqref{eq:general-g}.
The  problem of minimizing an expected loss---or, equivalently, of maximizing expected utility---under a budget constraint has  a long history; see, e.g., \cite{EkelandTemam} or, within a financial context, Section 3.3 in \cite{FS16}. In the existing results, the loss function is typically assumed to be continuously differentiable and often required to be strictly convex and to satisfy the Inada conditions (e.g., \citet[p.160]{FS16}). Let us point out that none of these assumptions is imposed here. This level of generality is crucial for us, because our subsequent Theorem \ref{th:r1-3} forms the basis for the proofs of the results in Section~\ref{rm section}. For instance, in Corollary \ref{cor:ES}, it will be applied to the loss function $\ell(x)=x_+=0\vee x$, which clearly satisfies none of the classical requirements.
In what follows, $\ell_+$ and $\ell_-$ are the positive and negative parts of $\ell$, respectively.

\begin{theorem}\label{th:r1-3} Suppose that Assumptions \ref{assm:g}, \ref{assm:gamma}, and \ref{assm:vw} hold. Let $\ell_+$ have growth index $q^+\in[1,\infty]$  and suppose that $w$ has growth index $p\in[1,\infty]$. If, moreover, $\ell$ is not bounded from below, let  $\ell_-$ have growth index $q^-\in[0,1]$  and suppose that the growth index $r$ of $v$ satisfies $r<\infty$ if $q^-=0$ and $r\le pq^+/q^-$ otherwise. Then the expected loss $\rho_\ell$ is robust against optimization at $X\in L^{pq^+}_n$ for $(\mathcal G,L^{pq^+}_n, \pi^{pq^+}_n)$.
\end{theorem}

If $\ell $ is bounded from below, our assumptions allow us to take $v\equiv-\infty$, in which case the lower bound is immaterial. If, on the other hand, both $v$ and $w$ are bounded, then the expected loss $\rho_\ell$ will be robust at $X\in L^0_n$ for $(\mathcal G,L^0_n,\pi^W_n)$.

Theorem \ref{th:r1-3} implies that 
the optimization of expected loss or expected utility is also generally robust under mild regularity conditions. 
This is  similar  to the optimization of convex risk measures (Theorems \ref{th:divergence risk measures} and \ref{th:usr risk measures}), and in sharp contrast to that of VaR (Theorem \ref{th:var1}).

\section{Simulation results}\label{sec:simulation} 

In this section, we illustrate the robustness and non-robustness of  $\ES$ and $\VaR$ against optimization by means of numerical simulations based on the formulas obtained in Propositions \ref{prop:r1-var1} and   \ref{pr:Sekine lemma1} for Problem \eqref{eq:r1-var1}. In our setup, the true risk factor $Z$ is either exponentially or Pareto distributed with an unknown parameter $\theta$.\footnote{The exponential($\theta$) distribution function $F$ with parameter $\theta >0$ is specified as $F(x)=1-e^{-\theta x}$, $x \ge 0$.}  The investor obtains an estimate $\widehat\theta$ for  $\theta$ and considers a corresponding model $X$. Then the investor minimizes $\rho(g(X))$ within the class of all measurable functions $g$ satisfying the inequality $0\le g(x)\le x$ and the budget constraint $\mathbb E[\gamma(X)X]\ge x_0$. We consider both $\rho=\VaR_{0.99}$ and $\rho=\ES_{0.975}$, where the respective levels 0.99 and 0.975 are chosen according to Basel III regulation (\cite{BASEL16}). To keep things simple, we let $\gamma(x)=x$; this choice allows us to compute several auxiliary quantities in closed form, thus reducing the possible impact  of numerical errors.

In Figure~\ref{ES_and_VaR_Pareto5}, we have seen that, in the case of Pareto-distributed risks,  the true risk $\VaR_{0.99}(g_X(Z))$ is substantially larger than the modeled risk  $\VaR_{0.99}(g_X(X))$, as soon as the model distribution underestimates the tail risk probabilities for $Z$. Figure \ref{ES_and_VaR_Exp1} establishes the same effect for exponentially distributed risks, thus showing that the issue persists for light-tailed risks. Taking $\widehat\theta=1$, we observe specifically that   $\VaR_{0.99}(g_X(X))=0.7720$ but $\VaR_{0.99}(g_X(Z))= 4.6098$ if $Z\sim\text{Exp}(0.999)$.  Note that a 0.1\% estimation error in the parameter $\widehat\theta$ leads here to   an increase of almost  500\% for the assessed risk.  As a matter of fact, Figure \ref{ES_and_VaR_Exp1} also shows that any benefits from optimizing $\VaR$ disappear as soon as $\theta<\widehat\theta$, because then   $\VaR_{0.99}(g_X(Z))$ becomes equal to the risk of the unoptimized position, $\VaR_{0.99}(Z)$. In sharp contrast, $\ES_{0.975}(g_X(Z))$ ranges within the narrow interval  $[1.14619,1.15216]$ for $\theta\in[0.1,1.5]$. That is,  the true risk $\ES_{0.975}(g_X(Z))$ deviates from the model value $\ES_{0.975}(g_X(X))$  by a mere $0.5\%$  as long as the true value $\theta$ is within $\pm50\%$ of the estimates value $\widehat\theta$. Throughout that entire interval of $\theta$-values, the optimized position $g_X(Z)$ provides a substantial and robust reduction of risk when compared with the ES of the non-optimized position,  $\ES_{0.975}(X)=4.6889$.\footnote{Since $\ES_{0.975}(Z)\approx \VaR_{0.99}(Z)$ for all $\theta>0 $ as studied by \cite{LW19}, we did not include the  plot of $\ES_{0.975}(Z)$  in  Figure \ref{ES_and_VaR_Exp1}.}

\begin{figure}[h]
\centering
\begin{overpic}[width=9cm]{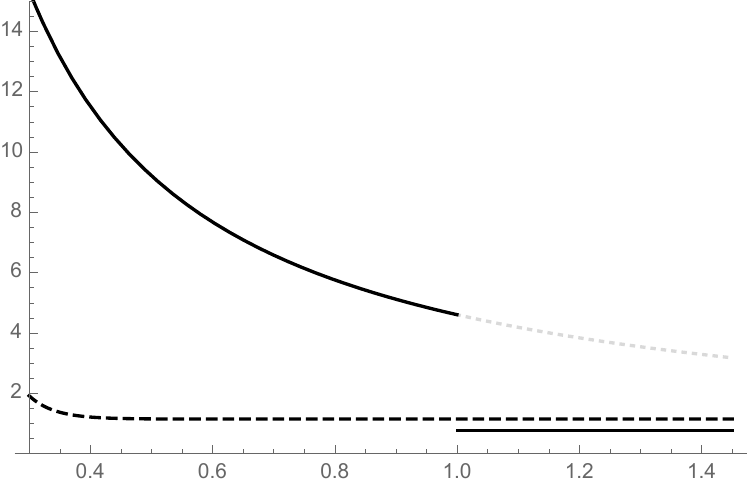}
\put(100,0){$\theta$}
\end{overpic}
\caption{
This plot shows $\rho(g_X(Z))$ for 
 $\rho=\VaR_{0.99}$ (solid) and $\rho=\ES_{0.975}$ (dashed), if $Z$ is exponentially distributed with parameter $\theta$. We assume that $X$ has an exponential distribution with parameter $\widehat\theta=1$. The dotted grey curve corresponds to the VaR of the unoptimized position, $\VaR_{0.99}(Z)$, which coincides with $\VaR_{0.99}(g_X(Z))$ for $\theta<\widehat\theta$.}
\label{ES_and_VaR_Exp1}
\end{figure}

The robustness of risk measures becomes important when risk measurement is combined with statistical estimation. This observation is at the core of the comparative discussion of the robustness of various risk measures; see e.g., \cite{CDS10}. In the following numerical experiment, we illustrate the impact of optimization on robustness. To this end, we refine the preceding setup by allowing for the statistical estimation of the parameter $ \theta$. That is, we generate $n$ iid realizations of $Z$ and   compute the maximum likelihood estimate $\widehat\theta$ from those realizations. Based on the estimated value of $\widehat\theta$, we compute the optimizing function $g_X$ and then compare the true risk value  $\rho(g_X(Z))$ to the perceived risk value $\rho(g_X(X))$. For each $n$, we repeat this procedure 10,000 times  and compute the mean-squared error, i.e., the average of $|\rho(g_X(Z))- \rho(g_X(X))|^2$,  of all 10,000 sample points. As the number $n$ of iid realizations of $Z$ increases, the estimate $\widehat\theta$ becomes ever more accurate, and we may expect the mean-squared error of the risk differences to decrease; this is indeed true for the case of ES, but not true for the case of VaR, because VaR is not robust against optimization. 
 Figure \ref{ParetoMSE fig} shows the corresponding mean-squared errors as a function of $n$ for the case of Pareto-distributed risks. Figure~\ref{ExpMSE fig} shows the analogous computations for exponentially distributed risks. Both figures illustrate that $\ES$ massively outperforms $\VaR$. 

 \begin{figure}[t]
 \centering
 \begin{minipage}{6.7cm}
 \begin{overpic}[width=6.7cm]{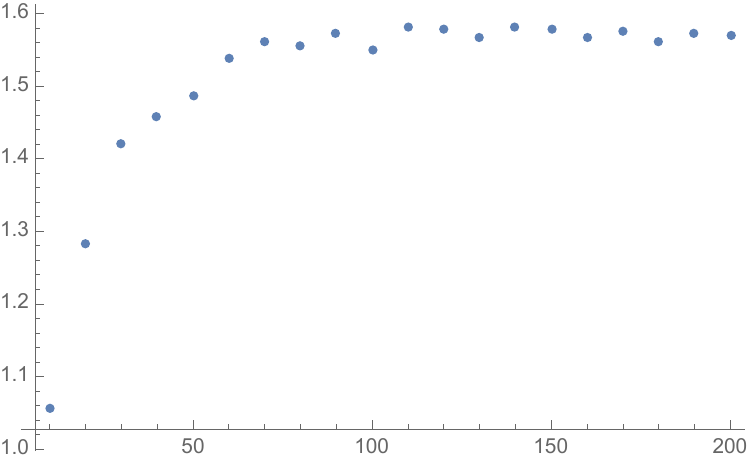}
 \end{overpic}
 \end{minipage}\qquad \begin{minipage}{6.7cm}
 \begin{overpic}[width=6.7cm]{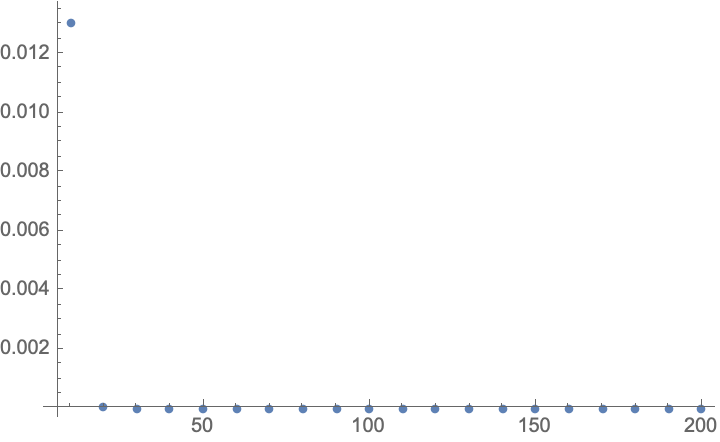}
 \end{overpic}
\end{minipage}\\
 \begin{minipage}{6.7cm}
 \begin{overpic}[width=6.7cm]{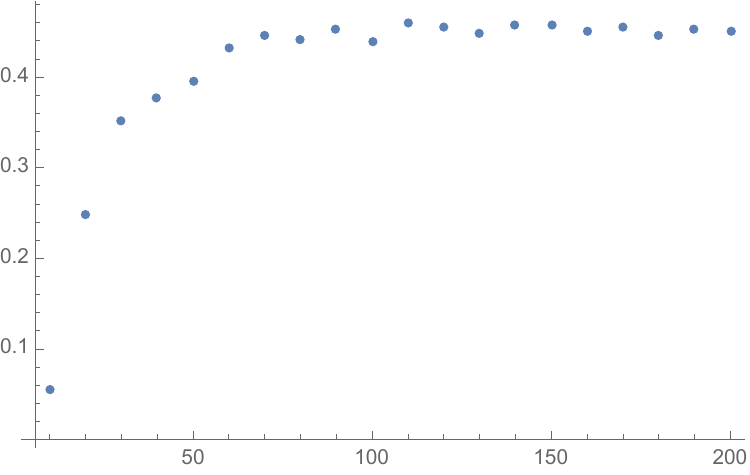}
 \end{overpic}
 \end{minipage}\qquad 
 \begin{minipage}{6.7cm}
  \begin{overpic}[width=6.7cm]{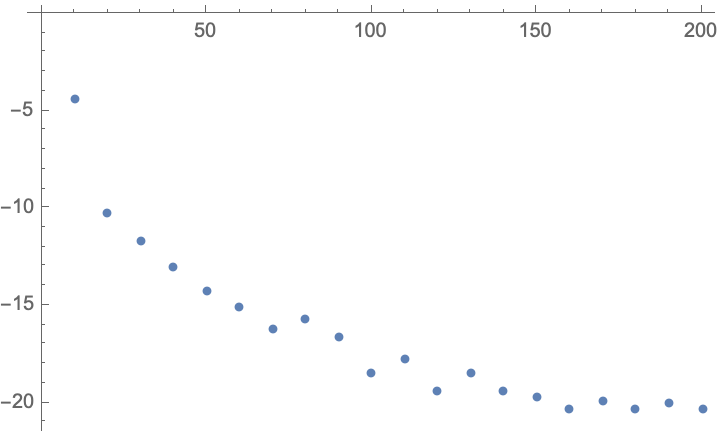}
 \end{overpic}
 \end{minipage}
 \caption{Mean-squared errors of 10,000 independent sample points of $\rho(g_X(Z))$ and $\rho(g_X(X))$, each with a   maximum likelihood estimator $\widehat\theta$ computed from $n$ iid~realizations of the Pareto(5)-distributed risk factor $Z$. The horizontal axis shows the number $n$. The case $\rho=\VaR_{0.99}$ can be found on the left, $\rho=\ES_{0.975}$ is on the right. The two top panels plot the mean-squared error (original values), the bottom ones correspond to their log-transforms. }\label{ParetoMSE fig} 
 \end{figure}

 \begin{figure}[t]
 \centering
 \begin{minipage}{6.7cm}
 \begin{overpic}[width=6.7cm]{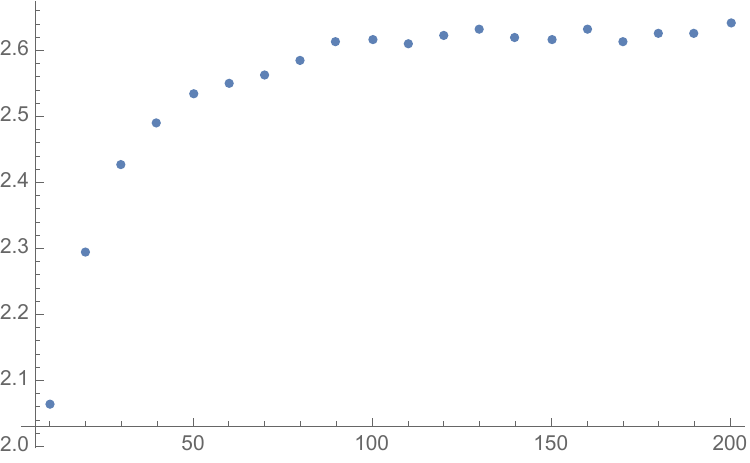}
 \end{overpic}
 \end{minipage}\qquad \begin{minipage}{6.7cm}
 \begin{overpic}[width=6.7cm]{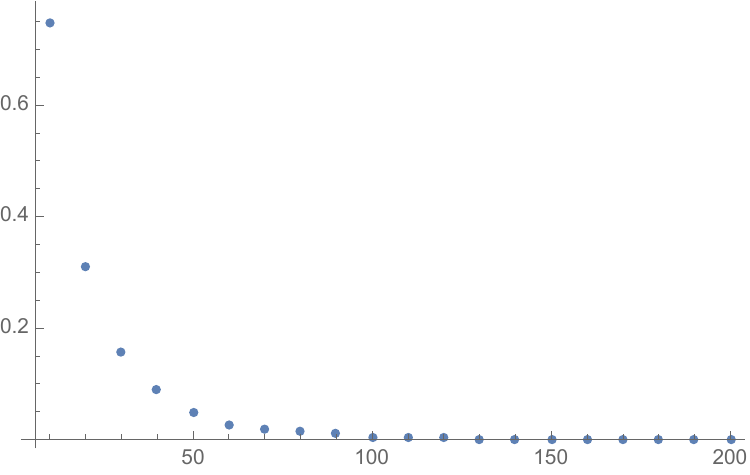}
 \end{overpic}
\end{minipage}\\
 \begin{minipage}{6.7cm}
 \begin{overpic}[width=6.7cm]{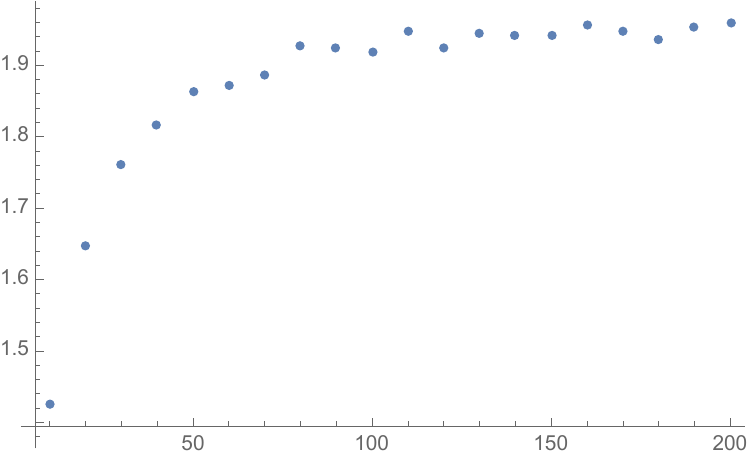}
 \end{overpic}
 \end{minipage}\qquad 
 \begin{minipage}{6.7cm}
  \begin{overpic}[width=6.7cm]{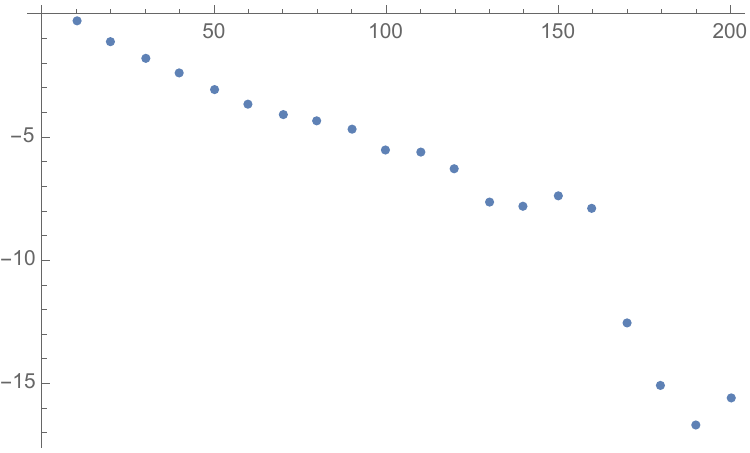}
 \end{overpic}
 \end{minipage}
 \caption{Mean-squared errors of 10,000 independent sample points of $\rho(g_X(Z))$ and $\rho(g_X(X))$, each with a  maximum-likelihood estimator $\widehat\theta$ computed from $n$ iid~realizations of the $\text{Exp}(1)$-distributed risk factor $Z$. The horizontal axis shows the number $n$. The case $\rho=\VaR_{0.99}$ can be found on the left, $\rho=\ES_{0.975}$ is on the right. The two top panels
  plot the mean-squared error (original values), the bottom ones correspond to their log-transforms}\label{ExpMSE fig} 
 \end{figure}

 \section{Discussions and remarks}\label{sec:7}
 \subsection{Implications of our results on regulatory risk measures}
  
In both the banking and the insurance sectors, VaR and ES are competing regulatory risk measures for solvency capital calculation; see, for instance, \cite{BASEL16} from the Basel Committee on Banking Supervision and \cite{IAIS14} from the International Association of Insurance Supervisors.
In this paper, with the new notion of robustness, we see that for the optimization problem \eqref{eq:general-opt}, VaR is generally not robust whereas ES is. This provides  strong support for the use of ES in optimization problems, in addition to its convexity which is very well recognized in the literature. These results   further support the transition from VaR to ES made by \cite{BASEL16} from a novel   theoretical perspective.

Our observations on the VaR vs.~ES issue can be explained intuitively.
From the proof of Theorem \ref{th:var1}, the VaR optimized positions always have a jump at the $p$-quantile level, and the optimized position can roughly be interpreted as a portfolio exhibiting a large loss with a small probability (e.g.~selling a large volume of far out-of-the-money call options).
This reflects the fact that ``VaR does not capture the tail risk" as indicated already by many academics and regulators (see e.g.~\cite{DEGKMRS01}, \cite{EPRWB14}, \cite{EKT15} and \cite{BASEL16}).  If there is model uncertainty around this $p$-quantile level, even if small, it ruins completely the optimality of the position.
This can explain the failure of the investment strategies (based on beliefs in small probabilities of default) of many of the larger banks before the 2008 financial crisis; see, for instance, the report by \cite{ACR10} on this matter.
We note that the optimized positions for VaR and ES may have   similar   forms (Proposition \ref{prop:r1-var1} and \ref{pr:Sekine lemma1}).
By definition, however, ES does not ignore the values of the tail part of the optimal allocation (in contrast to VaR), and this explains why the corresponding   value of the risk measure is not underestimated. 

There are extensive discussions on the robustness of  VaR and ES (although not in the context of optimization of this paper), and it may be fair to keep a balanced view. One important issue on the robustness of ES is the difficulty arising from perturbations  which yield probability distributions that may have infinite first moment; 
this is why in Corollary \ref{cor:ES},  
the robustness of ES  with respect to $\pi_n^q$ requires a condition on the growth rate of $v$ and $w$. Therefore, in the minimization of ES from historical data, one needs to always make suitable integrability assumptions, or otherwise minimizing ES  may be as problematic as the case of VaR. Infinite mean models are not of a purely academic nature in risk management; see for instance \cite{NEC2006} in the context of operational risk and \cite{W09} related to the economics of climate change.
For recent academic discussions on various  issues related to the desirability of VaR and ES in banking and insurance regulation, we refer to  \cite{KP16}, \cite{FZ16}, \cite{ELW18}, \cite{AB18} and the references therein.


\subsection{Remarks on the formulation of robustness}
 \label{sec:remarks} 

In this section we supply some further remarks on the relations between our definition of robustness and related notions in the literature on optimization and model uncertainty. 

  We start by considering the problem of solvency capital calculation of a firm  as set forth in the Basel III/IV and Solvency II agreements. Suppose that $Z$ is the true (but unknown) model and $g_Z\in {\mathcal G_Z(}\rho)$; see \eqref{eq:opt2} with $X=Z$. In   solvency capital calculation, the following quantities have different physical meanings: 
\begin{enumerate}[(a)]
\item $\rho(g_X(X))$: the \emph{perceived} risk value  (solvency capital requirement)   optimized for $X$ by the firm;
\item $\rho(g_Z(Z))$: the  \emph{idealistic} risk value optimized for $Z$ as if $Z$   {were} known;\ 
\item  $\rho(g_X(Z))$:  the \emph{actual} risk value    of the model $Z$, but the optimization is made for $X$.
\end{enumerate} 
Among the above quantities, the idealistic risk  value $\rho(g_Z(Z))$ represents what would be the best-case if the true model were known. Since the true model is not known, this value is not available  and hence irrelevant for the solvency capital calculation.
Therefore, for solvency risk management purposes, we are interested in
the \emph{solvency gap}
\begin{equation}
~~~~\underbrace{\rho(g_{{X}}({Z}))}_{\scriptsize \mbox{actual risk}} ~-~\underbrace{\rho (g_{{X}}({X}))}_{\scriptsize\mbox{ {perceived} risk}}, \label{eq:insolvency}
\end{equation}
not the 
\emph{optimality gap}
\begin{equation}
\underbrace{\rho(g_{ Z}( Z))}_{\scriptsize\mbox{idealistic optimum}}  ~-~\underbrace{\rho (g_{ X}( Z)) }_{\scriptsize\mbox{actual risk}}, \label{eq:optimality}
\end{equation}
nor the \emph{optimality shift} \begin{equation}
\underbrace{\rho(g_{ Z}( Z))}_{\scriptsize\mbox{idealistic optimum}}  ~~-\underbrace{\rho (g_{ X}( X)) .}_{\scriptsize\mbox{{perceived} optimum}}  \label{eq:shift}
\end{equation} 
Note that both \eqref{eq:optimality} and \eqref{eq:shift} involve $\rho(g_Z(Z))$ which is not relevant for solvency considerations. In the optimization literature, the continuity of the set mapping $Z\mapsto  {\mathcal G_Z(}\rho)$, as well as that of the function $Z\mapsto \rho(g_Z(Z)) $, is referred to as the problem of stability, i.e.~how do the optimal solutions and the optimality shift \eqref{eq:shift} change when the underlying model changes from $X$ to $Z$; see, e.g., \cite{BS00} and the references therein.

 {Let us further illustrate our notion of robustness by means of the following two examples.} 

\begin{example}
\label{ex:liquidation}
Suppose that the model $X$ leads to the unique optimal decision  $g_X(x)=x_0\in \R$, which means fully liquidating this asset or a perfect hedge.
In this case, $g_X$ is a constant function, and hence $Z\mapsto \rho(g_X(Z))$ is a constant mapping, thus always robust according to our definition.
In other words, the solvency gap \eqref{eq:insolvency} will be zero,  no matter what the optimizer for the true risk $Z$ is. Hence,  model uncertainty is irrelevant for the calculation of the solvency capital.
On the other hand, 
if the true model $Z$ is not equal to $X$ and liquidating the asset is not optimal for  $Z$, then we have $\rho(g_Z(Z))<\rho(x_0)$; thus
the optimality shift \eqref{eq:shift} will be strictly negative.
Therefore, the  solvency gap in \eqref{eq:insolvency} is the right notion to look at in this scenario, not   the optimality shift. 
\end{example}

 \begin{example}
\label{ex:identical}
Suppose that 
$X$ and $Z$ are similar  {in the sense} that 
$g_X(X)$ and $g_Z(Z)$ are identically distributed,  but $g_X(Z)$  {and $g_X(X)$ are} not identically distributed. 
In this case, we have $\rho(g_X(X)) = \rho(g_Z(Z))$ for any law-invariant risk measure $\rho$, such as $\VaR_p$ or $\ES_p$,
but  $\rho(g_X(Z))>\rho(g_Z(Z))= \rho(g_X(X))$ since $g_X$ is generally not optimal for $Z$.
  Clearly, the solvency gap \eqref{eq:insolvency}
is strictly positive 
and the optimality shift \eqref{eq:shift} vanishes. 
In this situation, the regulator is   concerned about the insolvency issue regarding model uncertainty. Indeed,  the true risk value $\rho(g_X(Z))$ is larger than the perceived risk value $\rho(g_X(X))$, which means that the solvency capital   is insufficient. Note however  that  there is no gap between  $\rho(g_Z(Z))$ and $\rho(g_X(X))$.
Therefore, also here,  the solvency gap is the right notion to study rather than the optimality shift. 
\end{example}


\begin{remark}\label{rem:r1-1} In the following  few remarks, we provide various other interesting comments on our formulation of robustness against optimization.\footnote{We thank an anonymous referee for discussions on these issues, and in particular, pointing out the relationship of our notion and the classic ones via the triangular inequality.} 
\begin{enumerate}
\item \textbf{Triangular inequality.}
 Note    the triangular inequality $$|\rho(g_X(Z)) - \rho(g_X(X))|\le |\rho(g_X(Z)) -\rho(g_Z(Z))| + |\rho(g_Z(Z)) -\rho(g_X(X))|.$$  
This inequality suggests that if both the optimality shift  in \eqref{eq:shift} and the optimality gap  in \eqref{eq:optimality} converges to $0$ as $Z\to X$ in $\pi$, then 
we   have the solvency gap converges to $0$, and thus robustness against optimization. However, the opposite does not hold, as illustrated  by the situation in Example  \ref{ex:liquidation}, where  robustness is always guaranteed, since the optimizer for $X$ is to fully liquidate the asset, although the optimality shift is non-zero.
Hence, the study of   robustness against optimization is not equivalent to the combined study of the optimality shift and stability.  In other words, continuity in both optimality gap and optimality shift is sufficient but not necessary for robustness against optimization. 

\item \textbf{Alternative ways to formulate robustness.}
There are some alternative ways to formulate the notion of robustness in Definition \ref{def:1}.
We discuss them and explain   the advantages of our formulation. 
\begin{enumerate}
\item One may use  uncertainty on the set of probability distributions  instead of that on the set of random vectors.  
There are a few advantages to consider the misspecification of the random vector rather than its distribution. First, our framework is general in the sense that there is no restriction to law-invariant risk measures or utility functions. For the notion of robustness studied in this paper, a probability measure does not need to be specified (non-law-invariant risk measures include  e.g., the margin requirement risk measure used by the Chicago Mercantile Exchange; see \citet[Section 2.3]{MFE15}).  Second, our framework is flexible as  we can  easily incorporate the misspecification of distribution by using a metric on the set of distribution (which is a pseudo-metric on the set of random variables). 
Third, in the proof of several results in the paper on the robustness and non-robustness of risk measures in Sections \ref{sec:VaR} and \ref{rm section}, we need to obtain equalities in the almost sure sense ($\omega$-wise  equalities) to locate the unique form of $g$.   
\item An alternative to Definition \ref{def:1} would be to require all, instead of one,  optimizing functions  $g_X$ to satisfy $\pi$-continuity of $Z\mapsto \rho(  g_X (Z))$ at $Z=X$. 
This requirement would be stronger than the current Definition \ref{def:1}.
For our main result, Theorem \ref{th:var1}, stating VaR is generally \emph{not robust},  
the current formulation in Definition \ref{def:1} gives a stronger result. 
Moreover, requiring continuity of all optimizing functions  may lead to pathological statements. 
For instance, suppose that $g_X\in \mathcal G$ is a continuous optimizing function for ES (or any other risk measure) and $X\in L^1$.  If one modifies $g_X$ on a set of $\mu_X$-measure zero (such as the set of rational numbers), then  the resulting function is still optimal, but robustness fails, and rightly so. 
\end{enumerate}
 
 \item \textbf{Limitations.} 
Robustness against optimization that we study in this paper is a desirable notion of robustness, but it should be seen as a necessary, but generally not sufficient, condition for being a good risk measure to use in the context of optimization.  
For instance, even with the continuity of $g_X(Z)$ at $Z=X$,   small-size perturbations in the model may lead to   enormous changes  in the risk assessment in a practical example, as our notion does not quantify sensitivity of the risk value, which will be a future research direction.

\item \textbf{Optimization over stochastic processes.}
 One can consider continuous-time models where optimizers are chosen over a set of stochastic processes (e.g.,~admissible trading strategies).
Our framework and discussions can be extended to such problems, as long as the optimizers are functions of the random source $X$, be it finite dimensional or infinite dimensional.
In fact, in many classic financial models, the continuous-time optimization problem (such as the hedging example above) can be translated into a single-period optimization problem via the martingale approach as we see Example \ref{ex:ctfinmarket}.
\end{enumerate}
\end{remark}

\section{A connection to distributionally robust optimization}\label{sec:6a}
 
We conclude this paper by discussing our notion of robustness in the context of \emph{distributionally robust optimization}.
Our results in Section \ref{sec:VaR} show that VaR is generally not robust for \eqref{eq:general-opt}.
In a classic setting of distributionally robust optimization (e.g.~\cite{QZ08}, \cite{ZF09}, \cite{BM18}),
the objective functional itself is evaluated under the worst-case value over a set of possible models representing uncertainty.  
By taking the worst-case value of the objective, model uncertainty is   incorporated into the optimization problem.  Some other relevant results on VaR and ES with the worst-case approach  can be found in \cite{HH13} and \cite{ZKR13}. 
We wonder  whether robustness against optimization of risk measures would be improved by taking such an approach.

To formulate this consideration mathematically, let $\rho$ be a compatible objective functional for an uncertainty triplet $(\G, \Z,\pi)$ and $X\in \Z$.
We look at the following optimization problem, which  is a   robust version of \eqref{eq:opt1},
\begin{equation}\label{eq:robopt}\mbox{to minimize: } \sup_{\pi({Y},X)\le \epsilon}\rho(g({Y})) \mbox{~~~~~subject to~} g\in \mathcal G,
\end{equation}
where  $\epsilon>0$.
Denote by $\mathcal G_X(\rho, \epsilon)$ the set of functions $g\in \mathcal G$ minimizing \eqref{eq:robopt}.
Clearly, if we allow $\epsilon=0$ in \eqref{eq:robopt}, then $\mathcal G_X(\rho, 0)= \mathcal G_X(\rho)$ and we are back in the setting of Section \ref{sec:2}.
In the problem \eqref{eq:robopt}, an investor is interested in the risk measure value $\rho(g(Z))$ of the risky position $g(Z)$, in which $Z$ is the unknowable true model.
Therefore,
similarly to Definition \ref{def:1}, we say that the objective functional $\rho$ is robust against optimization for the setting \eqref{eq:robopt} if there exists $g_X\in \mathcal G_X(\rho,\epsilon)$  such that
the function $Z\mapsto \rho(g_X(Z))$ is $\pi$-continuous at $Z=X$.

Unfortunately, the minimax  problem \eqref{eq:robopt} is not easy to solve analytically, even for   the representative settings in Section \ref{sec:3} and in the cases of VaR and convex risk measures.
Typically, a linear programming approach has to be applied for such problems.
As convex risk measures are already shown to be generally robust  in Section \ref{rm section},  its distributionally robust version is also generally robust;
we thus focus on the question of whether $\VaR$ becomes more robust in this context.
Our results in this section should be understood as exploratory rather than conclusive. 

To obtain analytic results, we look at a simple  one-dimensional case of \eqref{eq:general-opt}, by letting
\begin{equation}\label{eq:DRO-g}
\mathcal G=\{ g\in \mathcal G_1: 0\le  g  \le m,~ \E[\gamma(X) g( X)]\ge x_0\},
\end{equation}
where  $x_0$ and $m$ are two  constants satisfying $0\le x_0 < m \E[\gamma(X)]$.
We choose  $(\Z,\pi)=(L^\infty, \pi^\infty_1)$ and formulate the optimization problem
\begin{equation}\label{eq:robopt2}\mbox{to minimize: } \sup_{\pi^\infty_1({Y},X)\le \epsilon}\VaR_p(g({Y}))\mbox{~~over } g\in \mathcal G.
\end{equation}
Similarly to Section \ref{sec:VaR},
denote by $q$ the minimum of \eqref{eq:robopt2}, that is, $$q=\inf\left\{ \sup_{\pi^\infty_1({Y},X)\le \epsilon}\VaR_p(g({Y})): 0\le  g  \le m,~ \E[\gamma(X) g( X)]\ge x_0\right\}.$$
We make the following stronger assumption.

\renewcommand\theassumption{D}
\begin{assumption}\label{assm:6}
$q>0$, $1/2\le p<1$, $X$ has a  decreasing density on its support and $\gamma$ is an increasing function of $X$.\footnote{The monotonicity of $\gamma$ as a function of $X$ has a simple economic meaning. Recall that $X$ represents the loss of an asset. Hence, Assumption \ref{assm:6} requires that the pricing density is larger when the asset has a larger loss.
This requirement is satisfied by classic equilibrium models in the notion of Arrow-Debreu (\cite{AD54}).}
\end{assumption}

Fortunately,  with Assumption \ref{assm:6}, we are able to obtain an explicit form of the solution to Problem \eqref{eq:robopt2}, allowing us to compare the corresponding robustness property with the results we obtained in  Section \ref{sec:VaR}.

\begin{proposition}\label{prop:var-rob-new}
For $\mathcal G$ in \eqref{eq:DRO-g}, under Assumption \ref{assm:6},  Problem \eqref{eq:robopt2} admits a solution  of the form
  \begin{equation} \label{eq:varopt6-bd}
 g_X (x)=m\id_{\{x > c +\epsilon\}}  + q \id_{\{x  \le   c +\epsilon\}}, ~~x\in \R, \mbox{~
where $c=\VaR_{p}(X)$.}
\end{equation}
\end{proposition}

With the solution $ g_X $ in Proposition \ref{prop:var-rob-new}, the continuity of VaR mentioned in Section \ref{sec:3} implies that the mapping $Z\mapsto \VaR_p( g_X (Z))$ is $\pi^\infty_1$-continuous at $Z=X$.
As a consequence, $\VaR_p$ is robust against optimization for the setting \eqref{eq:robopt}.
This observation is in sharp contrast with Theorem   \ref{th:var1}, where we see that $\VaR_p$ is not   robust for $(\mathcal G,L^\infty,\pi^\infty_1)$ under some very weak assumptions (which does not conflict Assumption \ref{assm:6}).
Therefore, at least for the special setting \eqref{eq:robopt}, the modified optimization problem \eqref{eq:robopt2} improves the robustness of VaR.
It is unclear how this result can be generalized to other optimization problems, as analytic results for \eqref{eq:robopt} are rarely available.

Although $\VaR_p$ becomes robust in the setting \eqref{eq:robopt2}, its optimizing function takes a similar form as in Proposition \ref{prop:r1-var1}, that is, a distribution with a jump and a big loss with small probability.
Since the distribution of $g_X(X)$ in \eqref{eq:varopt6-bd} has a jump at its $(p+\epsilon)$-quantile, this type of optimizing functions is highly undesirable and is subject to considerable model uncertainty if $\epsilon$ is small; see the discussions in Section \ref{sec:7}.

 ~
 
\noindent{\textbf{Acknowledgements.}
The authors  thank 
Rama Cont, Xiaoxue Deng, Paul Glasserman, Liyuan Lin, Marcel Nutz and Philip Protter  for insightful comments on an early version of the paper.
In particular, the term ``robustness against optimization" was suggested by Paul Glasserman.
RW acknowledges financial support from the Natural Sciences and Engineering Research Council of Canada (NSERC, RGPIN-2018-03823, RGPAS-2018-522590)
 and from the Center of Actuarial Excellence Research Grant from the Society of Actuaries. 
 }

     \setcounter{lemma}{0}
     \renewcommand{\thelemma}{A.\arabic{lemma}} 
          \setcounter{proposition}{0}
     \renewcommand{\theproposition}{A.\arabic{proposition}}
               \setcounter{corollary}{0}
     \renewcommand{\thecorollary}{A.\arabic{corollary}}
               \setcounter{example}{0}
     \renewcommand{\theexample}{A.\arabic{example}} 
                    \setcounter{equation}{0}
     \renewcommand{\theequation}{A.\arabic{equation}} 
     
 \appendix

\section{Proofs of theorems and propositions}

\subsection{Proofs in Section \ref{sec:2}} \label{app:a1}

\begin{proof}[Proof of Proposition \ref{prop:31}]
It suffices to show that the function $Z\mapsto \rho(g(Z))$ is $\pi$-continuous. By definition, for any $X,Y\in \Z$,
$
\pi_{g(\Z)}(g(X),g(Y))= \pi (X,Y).
$
Thus the $\pi_{g(\Z)}$-continuity of $\rho$ is equivalent to the $\pi$-continuity of the function $Z\mapsto \rho(g(Z))$.
\end{proof}

\begin{proof}[Proof of Proposition \ref{prop:32}]
\begin{enumerate}[(i)]
\item
It suffices to show that,  as $k\to\infty$, $X_k\to X$ in $\pi_n^\infty$ implies that $g(X_k)\to g(X)$ in $\pi_1^\infty$.
This is a direct consequence of the  Heine-Cantor Theorem (see Theorem 4.19 of \cite{R76}).
\item $X_k\to X$ w.r.t.~$\pi^q_n$ implies that $\{|X_k|^q\}_{k\in \N}$ is uniformly integrable and that $X_k\to X$ in probability. It follows from the Continuous Mapping Theorem that $g(X_k)\to g(X)$ in probability. Moreover, for sufficiently large $c$,
$$\mathbb E\Big[|g(X_k)|^q\id_{\{|g(X_k)|>c\}}\Big]\le C^q\mathbb E\Big[|X_k|^q\id_{\{|X_k|>c/C\}}\Big].
$$
Therefore, $(|g(X_k)|^q)$ is uniformly integrable and, in turn, $g(X_k)\to g(X)$ w.r.t.~$\pi^q_1$.
\item It suffices to show that, as $k\to\infty$, $X_k\to X$ in $\pi_n^W$ implies that $g(X_k)\to g(X)$ in $\pi_1^W$.
This is a direct consequence of  the Continuous Mapping Theorem. \qedhere
\end{enumerate}
\end{proof}

\subsection{Proofs in Section \ref{sec:VaR}} \label{app:r1-1}

Since a rescaling of $\gamma$ does not change the optimization problem \eqref{eq:general-opt}, we will safely assume $\E[\gamma(X)]=1$ in the proofs of all results in Sections \ref{sec:VaR}-\ref{utility section}.

\begin{proof}[Proof of Theorem \ref{th:var1}]
In what follows,   equalities and inequalities on functions on $\R^n$ are understood as almost surely with respect to $\mu_X$,
and essential suprema and expectations of these functions are taken under $\mu_X$ (we use $\E_X$ to emphasize the expectation with respect to $\mu_X$).

There is nothing to show if the set $\mathcal G_X(\rho)$ of minimizers  is empty. 
Suppose that 
 $g_X\in \mathcal G $ is a minimizer  to the problem \eqref{eq:opt1}.
We will show that $Z\mapsto \rho(g_X(Z))$ cannot be continuous at $X$, which gives the statement in the theorem.

We first show that  the budget constraint is always binding, that is,
  \begin{equation}\label{eq:binding}
 \E_X[\gamma g'_X]=x_0 \mbox{~~for any optimizer $g'_X$ to  \eqref{eq:opt1}}.
 \end{equation}  
Suppose $\E_X[\gamma g'_X] > x_0$ for contradiction.
Denote by $\epsilon = \E_X[\gamma g'_X] - x_0$,
$v_0= \esssup(v)$,
and let $g''_X = (g'_X-\epsilon)\vee v$.
Since $v\le g''_X \le w$ and $\E_X[\gamma g''_X] \ge \E_X[\gamma g' _X]-\epsilon =  x_0$, we have $g''_X \in \mathcal G$.
Moreover,  since $v_0 < \VaR_p(g'_X(X))$ by Assumption \ref{assm:v}, we have 
\begin{align*}
\VaR_p(g''_X(X))  & \le \VaR_p((g'_X(X)-\epsilon)\vee v_0 ) \\ & = \VaR_p( g'_X(X)-\epsilon) \vee v_0 \\& = (\VaR_p(g'_X(X))-\epsilon)\vee v_0< \VaR_p(g'_X(X)).
\end{align*} 
This contradicts the optimality of $g'_X$.  
Hence,  \eqref{eq:binding} holds.

Consider the probability space $(\mathbb R^n,\mathscr{B}(\mathbb R^n),\mu_X)$, where $\mathscr{B}(\mathbb R^n)$ is the Borel $\sigma$-field. Following \cite{WZ20}, a set $A\in\mathscr{B}(\mathbb R^n)$ is called a $p$-tail event for a measurable function $h:\mathbb R^n\to\mathbb R$ and $p\in(0,1)$, if $\mu_X(A)=1-p$ and  $h(x)\ge h(x')$ for $\mu_X$-a.e.~$x\in A$ and $x'\in A^c$. The existence of the tail event in any atomless probability space is implied by Lemma A.3 of \cite{WZ20}, noting that Assumption \ref{assm:g} guarantees that $(\mathbb R^n,\mathscr{B}(\mathbb R^n),\mu_X)$ is an atomless probability space.

 Let $A$ be a $p$-tail event of $g_X$ so that  $\p(X\in A)  =\mu_X(A)= 1-p$.
By the definition of $\VaR_p$, we have 
\begin{equation}
\label{eq:var-eq1}
\VaR_p(g_X(X)) = \esssup (g_X  | A^c),
\end{equation}  
where $\esssup (g_X|A^c)$ is the essential supremum (with respect to $\mu_X$) of $g_X$ conditional on $A^c$. 
Define the function $\hat g=g_X \id_{A^c} + w \id_{A}$. 
It is clear that $\hat g\ge g_X$. Moreover, since $\hat g$ and $g_X$ only differ on the  tail event  $A$, \eqref{eq:var-eq1} implies that $\VaR_p(\hat g(X)) = \VaR_p(g_X(X))$. 
Hence, $\hat g$ is also a minimizer to \eqref{eq:opt1}. 
Note that if $\hat g\ne g_X$, then we have $\E_X[\gamma \hat g] >  \E_X[\gamma g_X] =x_0$.
By \eqref{eq:binding}, the budget constraint is always binding and we conclude that 
  $\hat g=g_X$. Thus,
$g_X\id _A = w  \id_{A}$. 
 
 Next, suppose by way of contradiction that the quantile function of $g_X(X)$ is continuous at $p$. 
If a $p$-tail event $A'$ of $w$ is $\mu_X$-a.s.~equal to $A$, then we have, using $g_X\id _A = w \id_{A} = w \id_{A'}$,  
$$\VaR_p(w  (X)) \le \lim_{q\downarrow p}  \VaR_q(w  (X)) 
=  \lim_{q\downarrow p}  \VaR_q(  g_X(X))  = \VaR_p(g_X(X)),$$
a contradiction to $\VaR_p(g_X(X)) < \VaR_p(w (X))$ in Assumption \ref{assm:v}. 
Hence, any $p$-tail event $A'$ of $ {w}$ satisfies $A'\ne A$.
 Since both sets have the same probability, the set $C:=A'\setminus A$ must be such that $\alpha:=\mu_X(C)\in (0,  1-p]$. 

Write $a= \VaR_p(g_X(X))$ and $b=\VaR_p(w(X))$.
For each $\delta\in (0,\alpha)$, let $C_\delta$ be a subset of $C$ such that $\mu_X(C_\delta)=\delta$, 
$A_\delta$ be a  ($p+\delta$)-tail event   of $ {g_X}$, and
$B_\delta := A\setminus A_\delta$.
Note that $w\ge b$ and $g_X\le a$ on $C_\delta$  since $C\subset A'\setminus A$. 
Hence, $w-  g_X\ge b-a>0$  on $C_\delta$.
Moreover, $C_\delta \cap B_\delta=\varnothing$.  
Since $\gamma >0$, 
 $\mu_X(\gamma \ge \epsilon)\to 1$ as $\epsilon\downarrow 0$. 
 As a consequence, we can freely choose $ C_\delta$ such that $\gamma$ is bounded away from $0$ on $  C_\delta$ for each $\delta \in (0,\alpha)$.   
 Since $\gamma$ is bounded from above, we can let
  $\ell ,u $ be two constants such that 
$0<\ell <\gamma$ on $ C_{\alpha/2}$
and 
$\gamma<u<\infty$ on $B_{\alpha/2}$.

Let 
$
g_\delta = a \id_{ B_\delta}  + w \id_{C_\delta} + g_X (1-\id_{B_\delta \cup C_\delta})$.
In other words, $ {g_\delta}$ is obtained by 
decreasing the value of $g_X$ to $a$ on the set $  B_\delta$  of probability $\delta$, and increasing its value to $w$ on the set $C_\delta$   also of probability $\delta$.  
Clearly, $v\le g_\delta \le w$. 
Note that $g_X \le a$ on $A^c$, which implies   
$g_\delta\le a$ on $ A^c \setminus C_\delta$.
Moreover, $g_\delta\le a$
on $B_\delta$. 
Therefore, 
 $$ {\p(g_\delta(X)\le a) \ge  \mu_X(B_\delta) + \mu_X(  A  ^c \setminus C_\delta)=1-p},$$ which gives 
$
\VaR_p(g_\delta(X)) \le   a.
$  
Since the quantile of $g_X(X)$ is continuous at $p$,    there exists $\delta_0\in (0,\alpha/2)$ such that $|g_X-a|<  (b-a)\ell /u$ on $  B_{\delta_0}$.  
Putting the above observations together,
we have 
\begin{align*}
\E[\gamma g_{\delta_0}] - \E[\gamma g_X]  & = \E[\gamma (g_{\delta_0}-g_X)\id_{C_{\delta_0}}] - 
\E[\gamma (g_X-g_{\delta_0})\id_{B_{\delta_0}}]  
\\&\ge (b-a) \E[\gamma \id_{C_{\delta_0}}] - \frac{\ell }{u}  (b-a) \E[\gamma \id_{B_{\delta_0}}] 
\\& \ge (b-a) \ell \delta_0 -  \frac{\ell }{u}   (b-a) u \delta_0 
>0.
\end{align*} 
The facts that $\VaR_p(g_{\delta_0}(X)) \le    a=\VaR_p(g_X(X))$
and $\E_X[\gamma g_{\delta_0}]  > \E_X[\gamma g_X] =x_0$ 
further guarantees that $g_{\delta_0}$ is an optimizer to \eqref{eq:opt1}. 
However, this contradicts the fact that any optimizer to \eqref{eq:opt1} needs to satisfy 
\eqref{eq:binding}.
This contradiction shows the desired conclusion that  the quantile function of $g_X(X)$ has a jump at $p$.

 For $\epsilon>0$, let $\tilde A_\epsilon=\{x\in \R^n: d(x,A)\le \epsilon\}$, where $d$ is the Euclidean distance. 
For each $y\in \tilde A _\epsilon$, let $f_\epsilon(y)$ be a Borel function which maps $y$ to one of its nearest point in $A$; see e.g., \cite{JR85} for the existence of the Borel selector. 
 Define the random variables $Z_\epsilon$ by 
 $$
 Z_\epsilon = f_\epsilon(X)\id_{\{X\in \tilde A  _\epsilon\}} + X \id_{\{X\in \tilde A_\epsilon^c\}}.
 $$
 Note that $\pi_n^\infty(Z_\epsilon, X)\le \epsilon$. Hence, $Z_\epsilon \to X$ as $\epsilon \downarrow 0$ in $\pi^\infty_n$, which is the strongest metric $\pi$ that we consider. Moreover,   by Assumption \ref{assm:g}, $X$ has positive density over its support which is a convex set,
which implies $\p(Z_\epsilon \in A)= \p(X\in \tilde A_\epsilon)>\p(X\in A)=1-p$.
Also note that if $Z_\epsilon \in A$, then $g_X(Z)  \ge \lim_{q\downarrow p}  \VaR_q(  g_X(X))  $.
 Hence, 
 $$\VaR_p(g_X(Z_\epsilon))  \ge \lim_{q\downarrow p}  \VaR_q(  g_X(X))  > \VaR_p(g_X(X)),$$
 showing that $Z\mapsto \rho(g_X(Z))$ is not $\pi$-continuous at $X$.  
\end{proof}

\begin{remark} \label{rem:var-assm}
The assumption $\esssup(v)<\VaR_p(X;\mathcal G)$ in Assumption \ref{assm:v} is not essential. 
As we see from the proof, 
this assumption is used to show two conditions.
First, it is used to show \eqref{eq:binding}; i.e., the budget constraint is binding.  
 Second, it is used to guarantee that $v\le g_\delta \le w$ for $\delta$ small enough.
 These two conditions are both natural and quite weak.
 On the other hand, the assumption that $\gamma$ is bounded above is only used to guarantee that 
 $\gamma<u $ for some $u>0$ on the set $B_{\delta}$ with probability $\delta\downarrow 0$. 
Note that if $g_X(X)$ is continuously distributed, then the set $B_{\delta}$ is $\mu_X$-a.s.~equal to
$
\{x\in \R^n: \VaR_p(g_X(X))< g_X(x)<\VaR_{p+\delta} (g_X(X))\}.
$
Hence, in this case it suffices to assume that $\gamma$ is bounded from above in any small neighbourhood of $\{x\in \R^n: g_X(x)=\VaR_p(g_X(X))\}.$  This assumption is practically always satisfied. 
\end{remark}

The following lemma is needed to show Proposition \ref{prop:r1-var1}.
\begin{lemma} \label{lem:varopt3-ns} Under Assumptions \ref{assm:g} and \ref{assm:v},  Problem \eqref{eq:r1-var1} admits at least one solution.
\end{lemma}

\begin{proof}[Proof of Lemma \ref{lem:varopt3-ns}]
Define $\mathbb{Q}$ through $\d\mathbb{Q}/\d\mathbb{P}=\gamma$ and let $\mu=\mathbb{Q}\circ X^{-1}$. The set $\mathcal G $ is then a uniformly integrable subset of $L^1(\mu)$. Let $\{g_n\}_{n\in \N}$ be a minimizing sequence for $\VaR_p$ in $\mathcal G $. By the Dunford-Pettis and
Eberlein-\v Smulian theorems (Theorems IV.8.9 and V.6.1 of \cite{DS58}), there exists a subsequence $\{g_{n_k}\}_{k\in \N}$ that converges weakly in $L^1(\mu)$ to some function $ g_X \in L^1(\mu)$.  Since $\mathcal G$ is convex and closed in $L^1(\mu)$, we get $ g_X \in \mathcal G$. Moreover, weak convergence in $L^1(\mu)$ implies clearly that the laws of $g_{n_k}(X)$ converge weakly to the one of $ {g_X(X)}$. But $\VaR_p$ is a left-hand quantile and hence  lower semicontinuous with respect to weak convergence (see, e.g., Exercise A.6.1 in \cite{FS16}). This proves that $ g_X $ is optimal.
%
%
 \end{proof}

\begin{proof}[Proof of Proposition \ref{prop:r1-var1}]
It is straightforward to check that $\VaR_p(g_X(X))\le q$.   
By Lemma \ref{lem:varopt3-ns}, Problem \eqref{eq:r1-var1} has at least one optimizer. 
Let $g\in \mathcal G$ be an optimizer to \eqref{eq:r1-var1}.
Since $\VaR_p(g(X))=q$ and $g(X)\le X$, we have, in the sense of $\mu_X$-a.s., 
$$
g(X) \le X \id_A + (X\wedge q )\id_{A^c},
$$
where $A$ is a $p$-tail event of $g(X)$. Clearly, by taking an equality in the above inequality will only increase $\E_X[\gamma g]$ while maintaining $\VaR_p(g(X))\le q$, and it does not affect the optimality of $g$. 
Moreover, as we have seen in the proof of Theorem \ref{th:var1}, the budget constraint is binding; this implies that we cannot strictly increase  $\E_X[\gamma g]$ while maintaining $\VaR_p(g(X))\le q$.
Therefore, it has to be $g(X) = X \id_A + (X\wedge q )\id_{A^c}$.
Note that 
\begin{align}
\E[\gamma(X) g(X)] & = \E[\gamma(X) X\id_A] + \E[\gamma(X) (X\wedge q) \id_{A^c}] 
\notag
\\ &= \E[\gamma(X) X] - \E[ \gamma(X) (X-q)_+ \id_{A^c}] = \E[\gamma(X) X] - \E[ Y_+ \id_{A^c}] .
\label{eq:r1-var4}
\end{align}
Maximizing the above term over $A$ satisfying $\p(A)=1-p$, it is clear that the maximum of $\E[\gamma(X) g(X)]$ 
is attained when $Y_+ $ takes its smallest values on $A^c$.
In other words, $A$ is a $p$-tail event of $Y_+$. 
Moreover,  we have $q<\VaR_p(X)$ by Assumption \ref{assm:v}. Hence, $Y=Y_+>0$ on $A$, and $A$ is also a $p$-tail event of $Y$. 
Using again the fact that the budget constraint is binding,
any function $g'$  that does not maximize $\E_X[\gamma g']$ with fixed $\VaR_p(g'(X))=q$ cannot be an optimizer.   
Since the $p$-tail event of $ Y_+$ is unique by $\p(Y\le \VaR_p(Y))=p$,
we know that $g=g_X$ in \eqref{eq:r1-varopt4-ns} is the unique $g$ with $\VaR_p(g(X))=q$ such that $\E_X [\gamma g]=x_0$. 
As a consequence, $g=g_X$ is
 the $\mu_X$-a.s.~unique solution to \eqref{eq:r1-var1}. 
 
Finally, we show \eqref{eq:r1-var3}.  
Since  $A$ is a $p$-tail event of $Y_+$, 
  we know that $A^c$ is a $(1-p)$-tail event of $- Y_+ $.
 Further, we have $\E[-  Y_+  |A^c] = \ES_{1-p} (-  Y_+)$ by  Lemma A.7 of \cite{WZ20}.
Using the fact that the budget constraint is binding and \eqref{eq:r1-var4},  we obtain
 \begin{align*}
 x_0 & =
\E[\gamma(X) g_X(X)] \\& = \E[\gamma(X) X] - \E[  Y_+  \id_{A^c}] 
\\& = \E[\gamma(X) X] + p\E[-  Y_+  |A^c] 
 =\E[\gamma(X)X] + p\ES_{1-p} (- Y_+).
\end{align*} This gives the desired equality \eqref{eq:r1-var3}.
\end{proof}

\subsection{Proofs in Section \ref{rm section}} \label{app:r1-2}

 \begin{proof}[Proof of Theorem \ref{th:divergence risk measures}] 
 %
 In a first step, we show the existence of a minimizer in $\mathcal G$. Clearly, $\mathcal G$ is nonempty.  Let $\{g_n\}_{n\in\mathbb N}$ be a sequence in $\mathcal G$ such that $\rho(g_n(X))$ converges to $\lambda:=\inf_{g\in\mathcal G}\rho(g(X))$.
Since $ g_n$ takes values between $v$ and $w$, a standard Koml\'os-type argument (e.g., Lemma 1.70 in \cite{FS16}), allows us to pass to a sequence  $\{\widetilde g_n\}_{n\in\mathbb N}$ of convex combinations of the $g_n$ such that $\widetilde g_n$ converges  $\mu_X$-a.s.~to some function $g_0$. 
Dominated convergence yields that 
\begin{equation}\label{budget constraint satisfied eq} 
\mathbb E  [\gamma (X)g_0 (X) ]=\lim_{n\uparrow\infty}\mathbb E  [\gamma(X)  \widetilde g_n (X) ]\ge \liminf_{n\uparrow\infty}\mathbb E _X [\gamma (X) g_n (X) ]\ge x_0.
\end{equation}
Hence, $g_0$ belongs to $\mathcal G$. 
The convexity of $\rho$ implies that $\rho(\widetilde g_n(X))$  converges to the infimum value $\lambda$. Moreover, since $\rho$ enjoys the Fatou property due to \eqref{divergence risk measure rep eq}, we have $\rho(g_0(X))\le\liminf_n\rho(\widetilde g_n(X))$. Therefore $g_X:=g_0$ is a minimizer in $\mathcal G$. 

Now we derive the structure of $g_X$. To this end, consider $\ell(x):=\sup_{y\ge0}(xy-\varphi(y))$.
Then $\ell$ is convex, nondecreasing, nonconstant, and finite on $\mathbb R$, because $\varphi$ has superlinear growth. Moreover, Theorem 4.2 in \cite{BenTalTeboulle2} (see also Lemma 1 in \cite{BenTalTeboulle1}) provides the following dual representation,
\begin{equation}\label{OCE dual rep eq}
\rho(Y)=\inf_{z\in\mathbb R}\big(\mathbb E[\ell(Y+z)]-z\big),\qquad Y\in L^\infty.
\end{equation}
We claim that the infimum in \eqref{OCE dual rep eq} is actually attained. To see why, note first that our assumptions on $\varphi$ imply that there exists $y_0>1$ such that $\varphi(y_0)<\infty$, which implies that the slope of $\ell$ is at least $y_0$ for sufficiently large $x$. This yields an upper bound on the range of all $z$ that  contribute to the infimum in \eqref{OCE dual rep eq}. Second, there is $y_1<1$ with $\varphi(y_1)<\infty$, which implies that the slope of $\ell$ is at most $y_1$ for sufficiently large negative $x$. This gives a lower bound on $z$. The continuity of $z\mapsto \mathbb E[\ell(Y+z)]-z$ now yields our claim.

Now let $z^*$ be such that $\rho(g_X(X))=\mathbb E[\ell(g_X(X)+z^*)]-z^*$. Then,
\begin{align*}
\inf_{g\in\mathcal G}\rho(g(X))&\le\inf_{g\in\mathcal G}\big(\mathbb E[\ell(g(X)+z^*)]-z^*\big)\le \mathbb E[\ell(g_X(X)+z^*)]-z^*=\inf_{g\in\mathcal G}\rho(g(X)).
\end{align*}
Hence, $g_X$ minimizes $\E[\ell(g(X)+z^*)]$ over $g\in\mathcal G$. We are thus in the context of Theorem \ref{th:r1-3} (optimizing expected loss), whose proof yields the form of $g_X$ as the minimizer of $\mathbb E[\ell(g(X)+z^*)]$ over $g\in\mathcal G$. Now suppose that $Z_n\in L^0_n$ are random variables whose laws converge to the one of $X$. By Skorokhod embedding, we may assume without loss of generality that $Z_n\to X$ holds $\mathbb P$-a.s.  The robustness of $\rho$ now follows as in the proof of Theorem \ref{th:r1-3} by using the fact that $\rho$ enjoys the so-called Lebesgue property, which in turn is a consequence of Exercises 4.2.3 and 4.3.4 in \cite{FS16}. \end{proof}

\begin{proof}[Proof of Corollary \ref{cor:ES}] First, by using the rightmost representation in \eqref{ES representations eq}, the existence of  a minimizer $g_X$ can be shown as in the proof of Theorem \ref{th:divergence risk measures} when replacing \eqref{budget constraint satisfied eq} with the following, more general argument. Our assumption that $\mathbb E[\gamma(X)|w(X)|]$ is finite yields with  Fatou's lemma that 
$$\mathbb E[\gamma(X)g_0(X)]\ge \limsup_{n\uparrow\infty}\mathbb E[\gamma(X)\widetilde g_n(X)]\ge \liminf_{n\uparrow\infty}\mathbb E[\gamma(X)g_n(X)]\ge x_0.
$$
Next, we will use the identity 
$$\ES_p(Y)=\min_{z\in\mathbb R}\bigg(\frac1{1-p}\mathbb E[(Y+z)_+]-z\bigg),\qquad Y\in L^1,
$$
where the minimum is attained at $z=F_Y^{-1}(1-p)$; see, e.g., Proposition 4.51 in \cite{FS16} and note that the proof given there works without modification  for $Y\in L^1$ and does not require the assumption $Y\in L^\infty$. Thus, we are in the setting of Theorem \ref{th:divergence risk measures}, and the remainder of the present proof follows exactly as for that result. All one needs to note in addition is that  $\ES_p$ is continuous on $L^1$.
\end{proof}

\begin{proof}[Proof of Theorem \ref{th:usr risk measures}] Let $\ell^*(y):=\sup_{x\in\mathbb R}(xy-\ell(x))$ be the Fenchel--Legendre transform of $\ell$. The risk measure $\rho$ can be represented in the form 
\begin{equation}\label{shortfall rep eq}\rho(Y)=\max_{\mathbb{Q}\ll \mathbb{P}}\bigg(\mathbb{E}_\mathbb{Q}[\,Y\,]-\inf_{\lambda>0}\frac1\lambda\bigg(x_0+
\mathbb{E}\Big[\,\ell^*\Big(\lambda\frac{d\mathbb{Q}}{d\mathbb{P}}\Big)\,\Big]\bigg)\bigg),\quad Y\in L^\infty;
\end{equation}
see  Theorem 4.115 in \cite{FS16}. Using this representation, the existence of a minimizer $g_X\in\mathcal G$ is established as in the proof of Theorem \ref{th:divergence risk measures}.

As shown at the beginning of the proof of Proposition 4.113 in \cite{FS16}, $z^*:=\rho(g_X(X))$ is the unique solution of the equation $\mathbb E[\ell(g_X(X)-z)]=x_0$. It follows from here that $g_X$ minimizes $\mathbb E[\ell(g(X)-z^*)]$ over $g\in\mathcal G$. Indeed, suppose by way of contradiction that there is $g_0\in\mathcal G$ for which $\mathbb E[\ell(g_0(X)-z^*)]<\mathbb E[\ell(g_X(X)-z^*)]$. Then the solution, $z_0=\rho(g_0(X))$, of the equation $\mathbb E[\ell(g_0(X)-z)]=x_0$ will be strictly smaller than $z^*$, a contradiction to the optimality of $g_X$.  The proof of Theorem \ref{th:r1-3} thus yields the structure of $g_X$ as a $\mu_X$-a.e.~continuous function. The robustness of $\rho$ now follows as in the proof of Theorem \ref{th:r1-3} by using the fact that $\rho$ enjoys the so-called Lebesgue property, which in turn is a consequence of Exercise 4.2.3 and Proposition 4.113 in \cite{FS16}.
\end{proof}

\begin{proof}[Proof of Corollary \ref{cor:expectile}] Let $\rho(Y)$ denote the expectile of $Y\in L^1$ and $\ell$ the convex loss function $\ell(x)=\tau x_+-(1-\tau) x_-$.  We have $\ell^*(y)=0$ if $1-\tau\le y\le \tau$ and $\ell^*(y)=+\infty$ otherwise. Hence, letting
\begin{equation}
\widetilde \rho(Y):=\sup\Big\{\mathbb E_{\mathbb Q}[Y]:\mathbb Q\ll\mathbb P\text{ and there exists $\lambda>0$ s.t. }1-\tau\le\lambda\frac{\d \mathbb Q}{\d\mathbb P}\le \tau\Big\},\quad Y\in L^1,
\end{equation}
the identity \eqref{shortfall rep eq} yields that $\rho(Y)=\widetilde\rho(Y)$ for $Y\in L^\infty$. For  $Y\in L^1$ and $n\in\mathbb N$, we let $Y_n:=(-n)\vee Y\wedge n$. Then we have $\rho(Y_n)=\widetilde\rho(Y_n)$. It is easy to see that $\rho(Y_n)\to\rho(Y)$. Moreover, the set of  densities of those probability measure $\mathbb Q\ll\mathbb P$ for which  there exists $\lambda>0$ such that  $1-\tau\le\lambda{\d\mathbb Q}/{\d\mathbb P}\le \tau$ is bounded in $L^\infty$. Therefore, Theorem 4.2 in \cite{CL09} implies that $\widetilde\rho$ is norm continuous on $L^1$. Consequently, $\widetilde \rho(Y_n)\to\widetilde\rho(Y)$, and we conclude that $\widetilde\rho(Y)=\rho(Y)$ for all $Y\in L^1$. Using this identity and the norm continuity of $\rho$, the robustness of $\rho$ now follows as in the proof of Theorem \ref{th:usr risk measures}. \end{proof}

\begin{proof}[Proof of Proposition \ref{pr:Sekine lemma1}]The result follows from  Theorem 8.26  of \cite{FS16}. The fact that  $r$ is a $p$-quantile of ${g_X(X)}$ is stated in the proof of Theorem 8.26  in \cite{FS16}.
\end{proof}

\subsection{Proofs in Section \ref{utility section}} \label{app:r1-3}

\begin{proof}[Proof of Theorem \ref{th:r1-3}] Consider the function
\begin{equation}\label{genoptvfctnmdefieq}
\ell^*(z,x):=\sup_{v(x)\le y\le w(x)} \big(yz-\ell(y)\big),
\end{equation}
defined for $y\in\mathbb{R}$ and $x\in\mathbb {R}^n$. Let $y^*(x,z)$ denote the largest maximizer. We must have
\begin{equation}\label{}
\begin{split}
y^*(x,z)=v(x)\quad&\iff\quad \ell^\prime_-(y)> z\quad\text{for all
$y\in\big(v(x),w(x)\big]$,}\\
y^*(x,z)=w(x)\quad&\iff\quad \ell^\prime_-(y)\le z\quad\text{for all
$y\in\big(v(x),w(x)\big]$.}
\end{split}
\end{equation}
Moreover, $\ell_-'(y^*(x,z))\le z\le \ell_+'(y^*(x,z))$ if $y^*(x,z)\in\big(v(x),w(x)\big) $ (see, e.g., Proposition A.9 (a) in \cite{FS16}). Letting $I(z):=\inf\{y:\ell'_-(y)>z\}=\sup\{y:\ell'_-(y)\le z\}$ denote the right-continuous generalized inverse function of $\ell'_-$,  we hence see that $y^*(x,z)=I(z)$ in the latter case. Altogether, we obtain that $y^*(x,z)=v(x)\vee I(z)\wedge w(x)$.

Let us define
\begin{equation}\label{gc solution eq}
g^{(c)}(x):=v(x)\vee I(c\gamma(x))\wedge w(x),\quad x\in\mathbb R^n,\ c\in\mathbb R.
\end{equation}
The function $I$ is nondecreasing and hence has at most countably many jumps, which form a $\mu_X\circ\gamma^{-1}$-nullset, due to our Assumption \ref{assm:gamma}. Due to our assumption $\mathbb E[\gamma(X)|w(X)|]<\infty$, we may apply the monotone convergence theorem, which yields that the function $c\mapsto\mathbb E[\gamma(X)g^{(c)}(X)]$   decreases continuously from $\mathbb E[\gamma(X)w(X)]>x_0$ to $K:=\mathbb E[\gamma(X)(v(X)\vee I(0)\wedge w(X))]$ as $c$ decreases from $+\infty$ to $0$. Let us first consider the case in which $K< x_0$. In this case, there is some $c^*>0$ for which $\mathbb E[\gamma(X)g^{(c^*)}(X)]=x_0$. We show now that $g_X:=g^{(c^*)}$ is optimal. Indeed, from \eqref{genoptvfctnmdefieq}
 and our definition of $g_X$, it is clear that for arbitrary $g\in\mathcal G$,
 \begin{equation}\label{ell*ineq}
 c^*\gamma(X)g_X(X)-\ell(g_X(X))=\ell^*(c^*\gamma(X),X)\ge c^*\gamma(X)g(X)-\ell(g(X)).
 \end{equation}
Taking expectations on both sides of \eqref{ell*ineq} and using that $\mathbb E[\gamma(X)g(X)]\ge x_0$ hence yields that $\mathbb E[\ell (g_X(X))]\le \mathbb E[\ell (g(X))] $, which is the desired optimality. 
 
 Let us now turn to the case in which $K\ge x_0$. To this end,
 consider $a:=\inf_y\ell'_-(y)\ge0$ and $b:=\sup_y\ell'_-(y)\in[ a,\infty]$. Then $I(z)=-\infty$ for $z<a$ and $I(z)=+\infty$ for $z\ge b$. Moreover, $I(a)=\lim_{z\downarrow a}I(z)$ is finite if and only if $\ell$ is linear  on $(-\infty,I(a)]$ with slope $a$. Since $K\ge x_0$ can only occur if $I(0)$ is finite and we clearly have $I(0)\le I(a)$, it follows that $\ell$ is  linear on $(-\infty,I(a)]$  and $I(0)=I(a)$. On the other hand, the slope of $\ell$ on $(I(a),\infty)$ will be greater than  $a$. Therefore, any function $g\in\mathcal G$ taking values greater than $v\vee I(a)$ with positive $\mu_X$-probability must be suboptimal, provided that we can solve the following auxiliary problem:
 \begin{equation}\label{expected loss aux problem 1 eq}
 \text{minimize $\mathbb E[\ell(g(X))]$ over  $g\in\mathcal G_n$ with $ v\le g\le v\vee I(a)\wedge w$ and $\mathbb E[\gamma(X)g(X)]\ge x_0$.}
 \end{equation}
 If $a=0$, so that $\ell$ is flat on $(-\infty,I(a)]$, then every $g$ satisfying the constraints in \eqref{expected loss aux problem 1 eq}
 will be optimal. For instance, we can take
 \begin{equation}\label{trivial gX eq}
 g_X:=f\quad\text{for}\quad f:=v\vee I(a)\wedge w.
 \end{equation}
 If $a>0$, then we let $h:=f-v$ and replace $g$ in \eqref{expected loss aux problem 1 eq}
 with $f-g$. Then \eqref{expected loss aux problem 1 eq}
 is equivalent to the auxiliary problem, \begin{equation}\label{expected loss aux problem 2 eq}
 \text{maximize $\mathbb E[g(X)]$ over $g\in\mathcal G_n$ with $0\le g\le h$ and $\mathbb E[\gamma(X)g(X)]\le K-x_0$.} 
 \end{equation}
 If $K=x_0$, this problem has only the trivial solution $g\equiv 0$, and so \eqref{trivial gX eq}  is clearly the $\mu_X$-a.s.~unique solution to \eqref{expected loss aux problem 1 eq}.  
 For $K>x_0$, we choose $c_0>0$ such that $\mathbb E[\gamma(X)\id_{\{\gamma(X)\le c_0\}}h(X)]=K-x_0$; this is possible, because, by way of the linearity of $\ell$ on $(-\infty,I(a)]$, our assumption that both $\mathbb E[\ell(v(X))]$ and $\mathbb E[\gamma(X)|w(X)|]$ are finite implies that $\mathbb E[\gamma(X)\id_{\{\gamma(X)\le c\}}h(X)]$ is a finite and continuous function of $c\in\mathbb R$. Now define $g^*:=h\id_{\{\gamma\le c_0\}}$ and take any other $g\in\mathcal G_n$ satisfying the constraints in \eqref{expected loss aux problem 2 eq}. Then we have $(\gamma-c_0)(g-g^*)\ge0$ and hence
 \begin{align*}
 0&\le \mathbb E[(\gamma(X)-c_0)(g(X)-g^*(X))]\le -c_0\big(\mathbb E[g(X)]-\mathbb E[g^*(X)]\big).
 \end{align*}
This shows that $g^*$ solves \eqref{expected loss aux problem 2 eq}. It follows that 
\begin{equation} \label{gX lin solution eq}
g_X:=f-g^*=(v\vee I(a)\wedge w)\id_{\{\gamma>c_0\}}+v\id_{\{\gamma\le c_0\}}
\end{equation}
 solves our original problem in case $K>x_0$.

To summarize, our original optimization problem admits a solution $g_X$ that has one of the forms \eqref{gc solution eq}, \eqref{trivial gX eq}, or \eqref{gX lin solution eq}.
With this minimizer at hand, we can now proceed to prove the asserted robustness. So suppose that $Z_k\to X$ in $L^{pq^+}_n$. Since the functions $v$, $w$, and $\gamma$ are continuous $\mu_X$-a.e.~and since $I$ has at most countably many discontinuities, we have  $\ell(g_X(Z_k))\to \ell(g_X(X))$ in $L^0$. Moreover,
\begin{align*} |\ell(g_X(Z_k))|& \le |\ell_-(v(Z_k))|+|\ell_+(w(Z_k)))|
\\& \le c_1(1+|Z_k|^{rq^-})+c_2(1+|Z_k|^{pq^+})\le c_3(1+|Z_k|^{pq^+}).
\end{align*}
It follows that the sequence $|\ell(g_X(Z_k))|$ is uniformly integrable and so  $\mathbb E[\ell (g_X(Z_k))]\to \mathbb E[\ell (g_X(X))]$. This is the asserted robustness. \end{proof}

\subsection{Proofs in Section \ref{sec:6a}}

To prove Proposition \ref{prop:var-rob-new}, we need the following two lemmas.
In what follows, we denote by $A_\epsilon=\{x\in \R:  |x-y|\le \epsilon~\mbox{for some~}y \in A\}$ for a set $A\subset\R$ and $\epsilon>0$.

\begin{lemma}\label{lem:rob-opt(i)}If $A\subset\R$ is either compact or an interval, then $A_\epsilon\in\mathscr B(\R)$ and $$\sup_{\pi^\infty_1({Y},X)\le \epsilon} \p ({Y}\in A)=\p(X\in A_\epsilon) .$$
\end{lemma}

\begin{proof} If $A$ is a compact set or an interval, then so is $A_\epsilon$, which proves $A_\epsilon\in\mathscr B(\R)$. Next, for any $Y\in L^\infty$ with $\pi^\infty_1({Y},X)\le \epsilon$, the condition
 $Y\in A$ implies $X\in A_\epsilon$ a.s.
 Therefore, $\p(Y\in A)\le \p(X\in A_\epsilon)$, leading to $\sup_{\pi^\infty_1({Y},X)\le \epsilon} \p ({Y}\in A)\le \p(X\in A_\epsilon) $.
 To show the opposite direction of the inequality, it suffices to take
 $Y=  f_A(X)\id_{\{X\in A_\epsilon\}} +X\id_{\{X\not \in A_\epsilon\}}$,
 where, for a compact set $A$,  $f_A(x)$ is a nearest point of $x$ in $A$ (to be precise, there can be two such nearest points; by taking $f_A(x)$ to be the lower of the two, $f_A$ becomes lower semicontinuous and, hence, measurable). In the case in which $A$ is a nondegenerate interval, we fix a point $a$ in the interior of the interval and  let
 $$f_A(x)=\begin{cases} a\vee(x-\epsilon)&\text{if $x\ge \sup A$,}\\
 a\wedge(x+\epsilon)&\text{if $x\le \inf A$,}\\
 x&\text{otherwise.}
 \end{cases}
 $$
 In both cases, $|Y-X|\le \epsilon$, and $\p(Y\in A)=\p(X\in A_\epsilon)$, leading to the desired result.
\end{proof}

\begin{lemma}\label{lem:rob-opt} Let $\epsilon>0$, $p\in [1/2,1)$, and suppose that $X$ satisfies Assumption \ref{assm:6}. If
 $A\subset\R$ is a compact set or an interval satisfying $\sup_{\pi^\infty_1({Y},X)\le \epsilon}\p(Y\in A)\le 1-p$, then
$$\p(X> \VaR_{p}(X) +\epsilon )\ge \p(X\in A).$$ \end{lemma}

\begin{proof} By letting
 $A^*=(\VaR_{p}(X) +\epsilon ,\infty)$, the assertion can be rewritten as $\p(X\in A^*)\ge \p(X\in A)$. By Lemma \ref{lem:rob-opt(i)},  we have
 \begin{equation}\label{eq A*eps}
 1-p=\p(X\in A^*_\epsilon)= \sup_{\pi^\infty_1({Y},X)\le \epsilon} \p ({Y}\in A^*).
 \end{equation}
 If $\p(X\in A_\epsilon)<1-p$, we can  enlarge $A$ to obtain $\p(X\in A_\epsilon)=1-p$.
Then  $x:=\inf (A_\epsilon)$ satisfies $x\le \VaR_p(X)$ since $\p(X\le x)\le 1-\p(X\in A_\epsilon)=p.$

We consider two cases separately. First, we assume $x>\essinf X$.
It is from the definition of $x$ that $\inf (A)=x+\epsilon$.
Hence, $(x,x+\epsilon)\subset A_\epsilon \setminus A$.
Also note that $\p(X\in (x,x+\epsilon))\ge \p(X\in (\VaR_p(X),\VaR_p(X)+\epsilon))$ since $X$ has a decreasing density and $x\le \VaR_p(X)$.
Therefore, we have
\begin{align*}
\p(X\in A)&= \p(X \in A_\epsilon)-\p(X\in A_\epsilon \setminus A)
\\ & \le  1-p-\p(X\in (x,x+\epsilon))
\\ & \le 1-p -  \p(X\in (\VaR_p(X),\VaR_p(X)+\epsilon)) = \p(X\in A^*).
\end{align*}
Next, we assume $x\le \essinf X$.
Since $p\in [1/2,1)$, we have $\p(X< \VaR_p(X)+\epsilon)) >p \ge 1-p$.
Because $\p(X\in A_\epsilon)=1-p$ and $x\le \essinf X$, there exists $x_0\in (x,\VaR_p+\epsilon)$ such that $x_0\not \in A_\epsilon$.
Let $x_1=\sup \{y<x_0: y\in A_\epsilon \}$. By the definition of $A_\epsilon$ and $x_1$, we have $x_1-\epsilon>x$  and $(x_1-\epsilon,x_1)\subset A_\epsilon \setminus A$.
Using a similar argument as in  the first case, we have
\begin{align*}
\p(X\in A)&= \p(X \in A_\epsilon)-\p(X\in A_\epsilon \setminus A)
\\ & \le  1-p-\p(X\in (x_1-\epsilon,x_1))
\\ & \le 1-p -  \p(X\in (\VaR_p(X),\VaR_p(X)+\epsilon)) = \p(X\in A^*).
\end{align*}
We conclude that, in both cases, $\p(X\in A^*)\ge \p(X\in A)$. \qedhere
\end{proof}

\begin{proof}[Proof of Proposition \ref{prop:var-rob-new}] 
Recall that $\mathcal G$ is given by \eqref{eq:DRO-g},   
and $\E_X[h]$ means $\E[h(X)]$ for any function $h$. 
Take an arbitrary $g\in \mathcal G$.
Denote by  $$a=\sup_{\pi^\infty_1({Y},X)\le \epsilon} \VaR_p(g({Y})),$$
and let $h$ be given by
\begin{equation}\label{eq:new-h}
h(x)=
m\id_{\{g(x)>a\}}  + a \id_{\{g(x)\le a\}},~~x\in \R.
\end{equation}
For all ${Y}\in L^\infty$ with $\pi^\infty_1({Y},X)\le \epsilon$, we have $\VaR_p(g({Y}))\le a$.
Therefore, $\p (g({Y})>a)\le 1-p $, which implies $\VaR_p(h({Y}))\le a$.
Thus,
$$\sup_{\pi^\infty_1({Y},X)\le \epsilon} \VaR_p(h({Y}))= a.$$
Note that $h\in \mathcal G$ since $m\ge h(x)\ge g(x)\ge 0$, $x\in \R$.
Therefore, for any
$g\in \mathcal G $, we can find some $h\in \G$ of the form \eqref{eq:new-h}, such that
$$\sup_{\pi^\infty_1({Y},X)\le \epsilon} \VaR_p(h({Y})) = \sup_{\pi^\infty_1({Y},X)\le \epsilon} \VaR_p(g({Y})).$$
As a consequence, it suffices to search for optimizers $h\in \mathcal G$ of the  form \eqref{eq:new-h}. Moreover, for such an $h$, we have $\E_X[\gamma  h ]=m\mathbb Q(g(X)>a)+a\mathbb Q(g(X)\le a)$, where $\mathbb Q$ is given by $\d\mathbb Q/\d \mathbb P=\gamma$.  Due to the inner regularity of the law $\mathbb Q\circ X^{-1}$, we can find, for any $a'>a$, a compact set $K\subset\{g(X)>a\}$ such that
$h'(x)=m\id_{\{x\in K\}}+a'\id_{\{x\in K^c\}}
$
satisfies $\E_X[\gamma h']\ge \E_X[\gamma h  ]\ge x_0$. Since $\p(Y\in K)\le \p(g(Y)>a)\le 1-p$ for all ${Y}\in L^\infty$ with $\pi^\infty_1({Y},X)\le \epsilon$, we  have
$$\sup_{\pi^\infty_1({Y},X)\le \epsilon}\VaR_p(h'({Y}))= a'.$$
Let us denote by $\mathscr{K}$ the class of all compact set $K\subset\mathbb R$ satisfying
 $\p(Y\in K)\le1-p$ for all ${Y}\in L^\infty$ with $\pi^\infty_1({Y},X)\le \epsilon$. The above argument shows that $\mathscr K$ is not empty.
 Define a function \begin{equation}\label{eq:new-h2}
h_{K}(x)=m\id_{\{x\in K\}}+a_K\id_{\{x\in K^c\}},~~x\in \R,
\end{equation}
where $a_K\in \R$ is such that $\E_X[\gamma h_K]= x_0$.
The existence of $a_K$ is guaranteed by $\p(X\in K^c)\ge p>0$.
Note that $0<a_K<m$ since $x_0<m$ and $q>0$.

  The preceding argument shows that it is  sufficient to construct a function $h^*$ such that
$\E[\gamma h^*(X)]=x_0$ and
\begin{equation}\label{eq:h*optimality condition}
\sup_{\pi^\infty_1({Y},X)\le \epsilon} \VaR_p(h^*({Y})) \le \sup_{\pi^\infty_1({Y},X)\le \epsilon} \VaR_p(h_K({Y}))=a_K,
\end{equation}
for all  $K\in\mathscr{K}$.
  We define $h^*$ by
$$
h^*(x)= m\id_{\{x>c+\epsilon\}}  + a^* \id_{\{x\le c+\epsilon\}},~~x\in \R,
$$
where $a^*\le m$ is such that  $\E_X[\gamma h^*  ]=x_0$. Now let $K\in\mathscr K$ be given and $h_K$ of the form \eqref{eq:new-h2}. We take $k\in\mathbb R$ such that $\p(X>k)=\p(X\in K)$.
Lemma \ref{lem:rob-opt} gives  $$\p(X>c+\epsilon)\ge \p(X\in K)= \p(X>k)$$ and hence  $k\ge c+\epsilon$. Moreover, since $\gamma$ is an increasing function of $X$,   the upper Hardy--Littlewood inequality, in the form of \citet[Theorem A.28]{FS16},  yields
$$x_0\le \E_X [\gamma h_K ]\le\E[\gamma(X) (m\id_{\{X>k\}}+a_K\id_{\{X\le k\}})]\le \E[\gamma(X) (m\id_{\{X>c+\epsilon\}}+a_K\id_{\{X\le c+\epsilon\}})].
$$
Our condition  $\E_X [\gamma h^* ]=x_0$ therefore yields $a_K\ge a^*$. Since, moreover, by construction,
$$\sup_{\pi^\infty_1({Y},X)\le \epsilon} \VaR_p(h^*({Y})) =a^*\quad\text{and}\quad  \sup_{\pi^\infty_1({Y},X)\le \epsilon} \VaR_p(h_K({Y}))=a_K,$$
we conclude that \eqref{eq:h*optimality condition} holds and that $h^*$ is hence a solution to Problem \eqref{eq:robopt2}.
\end{proof}

\subsection{Robustness of VaR and ES in an unbounded setting}
\label{app:unbounded}
For $\rho=\VaR_p$ or $\rho=\ES_p$, we consider the unbounded optimization problem
\begin{equation}
\label{eq:r1-unbounded}
\mbox{to minimize: } \rho(g(X)) \mbox{~~~~~subject to~} 
g\in \mathcal G_n,~~\E[ \gamma (X) g(X)]\ge x_0,
 \end{equation} 
 where $\gamma:\mathbb R^n\to(0,\infty)$ and $x_0\in\mathbb R$.
Problem \eqref{eq:r1-unbounded} corresponds to \eqref{eq:general-opt} with $w=\infty$ and $v=-\infty$.

\begin{proposition}\label{prop:r1-unbounded}
Assume that $X$ is continuously distributed, $\E[\gamma(X)]<\infty$,  and $p\in (0,1)$. 
\begin{enumerate}[(i)]
\item For $\rho=\VaR_p$, Problem \eqref{eq:r1-unbounded} has no solution.
\item For $\rho=\ES_p$,
Problem \eqref{eq:r1-unbounded} admits a solution  if and only if  
\begin{equation}\label{eq:es-cm-ex}
\esssup \gamma(X)\le \frac1{1-p}.
\end{equation}
Moreover, if \eqref{eq:es-cm-ex} holds,  a solution to \eqref{eq:r1-unbounded} is given by the constant function $  g_X (\cdot)=x_0$.
\end{enumerate}
\end{proposition}

\begin{proof}
Denote by $$\mathcal G_{\rm ub} = \{g\in \mathcal G_n: \E[ \gamma (X) g(X)]\ge x_0\}.$$  \begin{enumerate}[(i)]
\item
Let $A$ be a set such that $\p(X\in A)=1-p$. 
Write $\lambda = \E[\gamma(X) \id_{\{X\in A\}}]>0$.
For $d< x_0$, define the function $$g_{d}(x)= d+\frac{x_0-d}{\lambda} \id_{\{x\in A\}},~~x\in \R^n. $$
Clearly, $g_d(X)\in \mathcal G_{\rm ub}$ because
$\E[\gamma (X) g_{d}(X)] = d+  \frac{x_0-d }{\lambda}\E[\gamma (X) \id_{\{X\in A\}} ]= x_0.
$
On the other hand, $\VaR_p(g_{d}(X))=d.$ Letting $d\to -\infty$,
$$\VaR_p(X;\mathcal G_{\rm ub})=\inf\{\rho(\VaR_p(g(X)): g\in \mathcal G_{\rm ub}\} =-\infty,$$
and hence  \eqref{eq:r1-unbounded} does not have an optimizer. 

\item By the dual representation of $\ES_p$ in \eqref{ES representations eq}, we have
$$\ES_p(Y)=\sup_{B\in \mathcal B_p} \E[B Y] ~~\mbox{for all }Y\in L^1,$$
where $\mathcal B_p= \{B\in L^\infty: \E[B]=1,~0\le B\le \frac{1}{1-p}\}.$
If $\esssup \gamma(X)\le \frac1{1-p}$, then  $\gamma(X) \in \mathcal B_p$, and hence
for any $g\in \G_{\rm ub}$, $\ES_p(g(X))\ge \E[\gamma(X) g(X)]\ge x_0$.
Clearly, taking the constant function $ g_X (\cdot)=x_0$ we have $ g_X \in \mathcal G_{\rm ub}$
and $\ES_p({g_X(X)})=x_0$.
Therefore, $ g_X $ is a solution to Problem \eqref{eq:r1-unbounded}.

Next, assume $\esssup \gamma(X)> \frac1{1-p}$.
Denote by $y=\E[\gamma (X)\id_{\{\gamma(X)>\frac{1}{1-p}\}}]>0$ and $k=\ES_p (\id_{\{{\gamma(X)}>\frac{1}{1-p}\}})$.
Note that $k\le y$ because \begin{align*}
\ES_p\left(  \id_{\{{\gamma(X)}>\frac{1}{1-p}\}}\right)&= {\sup_{B\in \mathcal B_p} \E[B \id_{\{{\gamma(X)}>\frac{1}{1-p}\}} ]} \\& \le  {\frac{1}{1-p}\E[ \id_{\{{\gamma(X)}>\frac{1}{1-p}\}} ]}<{ \E[\gamma (X)  \id_{\{{\gamma(X)}>\frac{1}{1-p}\}}]} .
\end{align*}
For $\lambda>0$,
take $g_\lambda (x)={\lambda }\id_{\{\gamma(x)>\frac{1}{1-p}\}} - \lambda y + x_0$, $x\in \R^n$.
It is clear that $\E[\gamma(X) g_\lambda (X)] =\lambda y -\lambda y + x_0=x_0$, and hence $g_\lambda \in \mathcal G_{\rm ub}$.
We can calculate
$$\ES_p(g_\lambda(X))- \E[\gamma(X) g_\lambda(X)] =
\lambda\left( \ES_p\left(\id_{\{\gamma(X)>\frac{1}{1-p}\}}\right) - y \right)=\lambda (k-y).
 $$
Letting $\lambda \to \infty$, we get
$$
\inf\{\ES_p(g(X)): g\in \mathcal G_{\rm ub}\}=-\infty,
$$
and hence there is no solution to Problem \eqref{eq:r1-unbounded}. \qedhere
\end{enumerate}
\end{proof}

As a direct consequence of Proposition \ref{prop:r1-unbounded},
 for any choice of   $(\mathcal Z, \pi)$,
 $\VaR_p$ is not robust against optimization for $(\mathcal G_{\rm ub},\mathcal Z,\pi)$, and if \eqref{eq:es-cm-ex}  holds, then $\ES_p$ is robust against optimization for $(\mathcal G_{\rm ub},\mathcal Z,\pi)$.  Clear from the proof, the assumption that $X$ has a continuous distribution is only used in part (i), and it can be relaxed to requiring $\p(X\in A)\in (0,1-p)$ for some event $A$.

  \end{document}